\documentclass[10pt, conference]{IEEEtran}
\IEEEoverridecommandlockouts
\usepackage{cite}
\usepackage{amsfonts}
\usepackage{amssymb}
\usepackage[cmex10]{amsmath}
\usepackage{mathrsfs}
\usepackage{ifpdf}
\ifCLASSINFOpdf
  \usepackage[pdftex]{graphicx}
\else
  \usepackage[dvips]{graphicx}
\fi
\usepackage[font=footnotesize]{subfig}
\usepackage{wrapfig}
\usepackage{multirow}
\usepackage{algorithmic}
\usepackage[lined, ruled, linesnumbered]{algorithm2e}
\usepackage{stfloats}
\usepackage{mdwmath}
\usepackage{mdwtab}
\usepackage{url}
\usepackage{array}
\usepackage{setspace}

\newtheorem{theorem}{Theorem}
\newtheorem{lemma}{Lemma}
\newtheorem{corollary}{Corollary}
\newtheorem{example}{Example}
\newcommand{\myspan}{\textnormal{span}}
\newcommand{\myrank}{\textnormal{rank}}
\newcommand{\mygcd}{\textnormal{gcd}}
\newcommand{\mylcm}{\textnormal{lcm}}
\newcommand{\myqed}{\hfill $\blacksquare$}

\newcommand{\NA}{PBNA}
\abovedisplayskip
\belowdisplayskip
\usepackage{color, verbatim}

\newcommand{\B}[1]{\mathbf{#1}}
\newcommand{\BF}{\mathbb{F}}
\newcommand{\C}[1]{\mathcal{#1}}
\newcommand{\SR}[1]{\mathscr{#1}}
\newcommand{\DEF}{\triangleq}
\newcommand{\BB}[1]{\mathbb{#1}}

\begin{document}

\title{On the Feasibility of Precoding-Based Network Alignment for Three Unicast Sessions \vspace{-0.3cm}}

\author{\IEEEauthorblockN{Chun Meng, Abinesh Ramakrishnan, Athina Markopoulou, Syed Ali Jafar}
\IEEEauthorblockA{ 
Department of Electrical Engineering and Computer Science \\
University of California, Irvine \\
Email: \{cmeng1, abinesh.r, athina, syed\}@uci.edu
}
\thanks{This work was supported by the NSF CAREER award (0747110) and by
an AFOSR MURI (FA9550-09-1-0643).}
}
\maketitle

\begin{abstract}
We consider the problem of network coding across three unicast sessions over a directed acyclic graph, when each session has min-cut one. 
Previous work by Das et al. adapted a precoding-based interference alignment technique, originally developed for the wireless interference channel, specifically to this problem. We  refer to this approach as precoding-based network alignment (PBNA). 
Similar to the wireless setting, PBNA asymptotically achieves half the minimum cut; different from the wireless setting, its feasibility depends on the graph structure.
Das et al. provided a set of feasibility conditions for PBNA with respect to a particular precoding matrix. However, the set consisted of an infinite number of conditions, which is impossible to check  in practice. Furthermore, the conditions were purely algebraic, without interpretation with regards to the graph structure. 
In this paper, we first prove that the set of conditions provided by Das. et al are also necessary for the feasibility of PBNA with respect to {\em any} precoding matrix. 
Then, using two graph-related properties and a degree-counting technique, we reduce the set to just four conditions. 
This reduction enables an efficient algorithm for checking the feasibility of PBNA on a given graph.
\end{abstract}

\section{Introduction}
Network coding was originally introduced to maximize the rate of a single multicast session over a network \cite{Ahlswede2000}\cite{Koetter2003}\cite{TraceyHo2006}. 
However, network coding across different sessions, which includes {\em multiple unicasts}  as a special case, is a well-known open problem. 
For example, finding linear network codes for multiple unicasts is NP-hard \cite{April2004}. Thus, suboptimal, heuristic approaches, such as linear programming \cite{Traskov2006} and evolutionary approaches \cite{MinkyuKim2009}, are typically used. 
Moreover, while it has been shown that scalar or vector linear network codes might be insufficient to achieve the optimal rate \cite{Medard2003}, only approximation methods \cite{Harvey2006} exist to characterize the rate region for this setting. 

In this paper,  we consider the simplest inter-session linear network coding scenario:  three unicast sessions over a directed acyclic graph, each session with minimum cut one. 
Das et al.  \cite{Das2010} applied a precoding-based interference alignment technique, originally developed by Cadambe and Jafar \cite{Cadambe2008} for wireless interference channel, to this problem; we refer to this technique as \textit{precoding-based network alignment (PBNA)}. 
In a nutshell, PBNA (i) simulates a wireless channel through random network coding \cite{TraceyHo2006} in the middle of the network and (ii) applies interference alignment at the edge, i.e., via precoding at the sources and decoding at the receivers. 
This way, it greatly simplifies the network code design, while it guarantees that each unicast session asymptotically achieves a rate equal to half of its minimum cut\cite{Das2010}.

An important difference from the wireless interference channel is that, in our problem, there may be dependencies between elements of the transfer matrix introduced by the graph structure, which make PBNA infeasible in some networks \cite{Ramakrishnan2010}. 
As a first step, Das et al. \cite{Das2010} provided a set of feasibility conditions for PBNA, and proved they are sufficient for the feasibility of PBNA with respect to a particular precoding matrix.
One important limitation is that the set consists of an infinite number of conditions, which makes it impossible to check in practice. 
Another limitation is the lack of consideration of graph structure, which turns out to be the reason for the  significant redundancy in the set of conditions.  
Ramakrishnan et al. \cite{Ramakrishnan2010} conjectured that the infinite set of conditions can be reduced to just two conditions. Han et al. \cite{Han2011}  proved that the conjecture holds for three symbol extensions; however, this result cannot be generalized beyond three symbol extensions.

In this paper, we make the following contributions. 
First, we prove that the set of conditions provided in \cite{Das2010} are also necessary for the feasibility of PBNA with respect to {\em any} valid precoding matrix. 
Then, using a simple degree-counting technique and two graph-related properties, we greatly reduce the set to just three conditions; two of them turn out to have an intuitive interpretation in terms of graph structure. 
Finally, we present an efficient algorithm for checking the three conditions.

The rest of this paper is organized as follows. 
In Section II, we present the problem formulation. 
In Section III, we summarize our main results. 
In Section IV, we discuss the graph-related properties that are key to the simplification of the conditions. 
In Section V, we prove and discuss our main results regarding the  feasibility condition of PBNA. 
In Section VI, we present an algorithm for checking the condition. 
In Section VII, we conclude the paper.
The Appendices provide details on the proofs that were outlined or omitted from the main part of the paper.

\section{Problem Formulation\label{sec_problem}}
The network is a delay-free directed acyclic graph, denoted by $G=(V,E)$, where $V$ is the set of nodes and $E$ the set of edges.
Without loss of generality, each edge has capacity one, i.e., can transmit one symbol of finite field $\BF_{2^m}$ in a unit time.
For the $i$th unicast session ($i\in\{1,2, 3\}$), let $s_i$ and $d_i$ be its sender and receiver respectively, and $R_i$ its transmission rate. 
Every edge $e\in E$ represents an error free channel. 
We assume that the minimum cut between $s_i$ and $d_i$ is one. 
Let $X_i$ be the source symbol transmitted at $s_i$ and $Z_i$ be the symbol received at $d_i$.
We further extend $G$ as follows: For the $i$th unicast session ($i\in\{1,2,3\}$), we add a virtual sender $s'_i$ and a virtual receiver $d'_i$ and two edges $\sigma_i=(s'_i,s_i)$ and $\tau_i=(d_i,d'_i)$. 
The extended graph is denoted by $G'=(V',E')$. 
For $e\in E'$, let $head(e)$ and $tail(e)$ denote its head and tail respectively.
 
In the middle of the network, we employ random network coding \cite{TraceyHo2006} to mimic wireless channel. 
The symbol transmitted along $e\in E'$, denoted by $Y_e$, is a linear combination of incoming symbols at $tail(e)$.
\begin{align*}
Y_e=\begin{cases}
X_i & \text{If } e=\sigma_i; \\
\sum_{head(e')=tail(e)} x_{e'e}Y_{e'} & \text{Otherwise.}
\end{cases}
\end{align*}
where $x_{e'e}$ is a variable, which takes values from $\BF_{2^m}$ and represents the coding coefficient used to combine the incoming symbol along $e'$ into the symbol along $e$. 
We group all coding coefficients $x_{e'e}$'s into a vector $\B{x}$, called the coding vector of $G'$. 
The network acts as a linear system: the output at $d'_i$ is a mixture of source symbols, $Z_i = \sum^3_{j=1}\nolimits m_{ij}(\B{x})X_j$, where $m_{ij}(\B{x})\in \BF_{2^m}[\B{x}]$ is the \textit{transfer function} from $s'_j$ to $d'_i$ and can be written as follows \cite{Koetter2003}:    
\begin{align*}
m_{ij}(\B{x})=\sum_{P\in \C{P}_{ij}}\nolimits t(P)
\end{align*}
where $\C{P}_{ij}$ is the set of paths from $s'_j$ to $d'_i$, and $t(P)$ is the product of coding coefficients along path $P$. 
We assume that all $m_{ij}(\B{x})$'s are non-zeros, which is the most challenging case. 
Indeed, as shown in Section \ref{sec_cond_na}, when some $m_{ij}(\B{x})$ $(i\neq j)$ is zero, the feasibility condition of PBNA is significantly simplified due to reduced number of interferences.

At the edge of the network, we apply interference alignment \cite{Das2010}\cite{Cadambe2008} via precoding at senders and decoding at receivers. 
Let $\B{X}_i=(X^1_i,\cdots,X^{k_i}_i)^T$ denote the input vector at sender $s'_i$, where $k_i$ is a two-phase function of some integers $n$, depending on whether $i$ equals one:
\begin{flalign*}
k_i = \begin{cases}
L_1(n) & \text{if } i=1 \\
L_2(n) & \text{otherwise}.
\end{cases}
\end{flalign*}
where $L_1: \BB{Z}^+ \rightarrow \BB{Z}^+$ and $L_1: \BB{Z}^+ \rightarrow \BB{Z}^+$ are two functions defined on $\BB{Z}^+$.
We will determine $L_1(n)$ and $L_2(n)$ later in this section. 
In order for PBNA to work properly, we require $L_1(n)$ and $L_2(n)$ satisfy the following condition:
\begin{flalign}
& L_1(n) \ge L_2(n) \label{eq_na_dim_1} \\
& \lim_{n \rightarrow \infty} \frac{L_1(n)}{L_2(n)} = 1
\end{flalign}
Define $L(n) = L_1(n) + L_2(n)$. 
As we will see later, the above two conditions are essential in the construction of a valid solution to PBNA.
We use {\em precoding matrix} $\B{V}_i$ to encode $\B{X}_i$  into $L(n)$ symbols, which are then transmitted via $L(n)$ uses of the network (time slots). The output vector at $d'_i$ is 
\begin{align*}
\B{Z}_i =(Z^1_i, \cdots, Z^{L(n)}_i)^T = \sum^3_{j=1}\nolimits \B{M}_{ij}\B{V}_j\B{X}_j
\end{align*}
where $\B{M}_{ij}$ is a $L(n)\times L(n)$ diagonal matrix with the $(k, k)$ element being $m_{ij}(\B{x}^k)$, where $\B{x}^k$ represents the coding vector for the $k$th use of the network.
$\B{V}_1$ is a $L(n)\times L_1(n)$ matrix, and $\B{V}_2, \B{V}_3$ are both $L(n)\times L_2(n)$ matrices. 
$\B{V}_i$ can still contain indeterminate variables. Let $\xi$ denote the vector of all variables in $\B{x}^1,\cdots,\B{x}^{L(n)}$ and $\B{V}_1, \B{V}_2,\B{V}_3$. 
We require the following conditions are satisfied for some values of $\xi$ \cite{Cadambe2008}:
\begin{align*}
& \SR{A}_1: \, \myspan(\B{M}_{12}\B{V}_2) = \myspan(\B{M}_{13}\B{V}_3) \\
& \SR{A}_2: \, \myspan(\B{M}_{23}\B{V}_3) \subseteq \myspan(\B{M}_{21}\B{V}_1) \\
& \SR{A}_3: \, \myspan(\B{M}_{32}\B{V}_2) \subseteq \myspan(\B{M}_{31}\B{V}_1) \\
& \SR{B}_1: \, \text{rank}(\B{M}_{11}\B{V}_1 \hspace*{8pt} \B{M}_{12}\B{V}_2) = L(n) \\
& \SR{B}_2: \, \text{rank}(\B{M}_{21}\B{V}_1 \hspace*{8pt} \B{M}_{22}\B{V}_2) = L(n) \\
& \SR{B}_3: \, \text{rank}(\B{M}_{31}\B{V}_1 \hspace*{8pt} \B{M}_{33}\B{V}_3) = L(n)
\end{align*}

Condition $\SR{A}_i$ guarantees that all the interferences at $d'_i$ are aligned, i.e., mapped into the same linear space, while condition  $\SR{B}_i$ ensures that all source symbols for the $i$th unicast session can be decoded. 
These conditions ensure that we can achieve a rate tuple $(R_1,R_2,R_3)=\B{R}_n\DEF(\frac{L_1(n)}{L(n)}, \frac{L_2(n)}{L(n)}, \frac{L_2(n)}{L(n)})$, which approaches $(\frac{1}{2}, \frac{1}{2}, \frac{1}{2})$ as $n\rightarrow \infty$. 
In this case, we say that $\B{R}_n$ is \textit{feasible through PBNA}. \footnote{In this paper, we first consider the feasibility conditions of PBNA for a fixed value of $n$. Then, in the Main Theorem, we prove that the feasibility conditions of PBNA are actually irrelevant to $n$ for $n>1$.}

Previous work \cite{Das2010}\cite{Ramakrishnan2010}\cite{Han2011} only considered the feasibility of PBNA under a particular precoding matrix, i.e., $\B{V}^*_1$ in Eq. (\ref{eq_v1}), which was first introduced in \cite{Cadambe2008}.
To address this limitation and characterize the feasibility of PBNA for any precoding matrix, we reformulate $\SR{A}_1, \SR{A}_2, \SR{A}_3$ and $\SR{B}_1, \SR{B}_2, \SR{B}_3$ without any assumption about the structure of precoding matrix. 
First, we reformulate $\SR{A}_1, \SR{A}_2, \SR{A}_3$ as:
\begin{flalign*}
& \SR{A}'_1:\, \B{M}_{12}\B{V}_2 = \B{M}_{13}\B{V}_3\B{A} \\
& \SR{A}'_2:\, \B{M}_{23}\B{V}_3 = \B{M}_{21}\B{V}_1\B{B} \\
& \SR{A}'_3:\, \B{M}_{32}\B{V}_2 = \B{M}_{31}\B{V}_1\B{C}
\end{flalign*}
where $\B{A}$ is an $L_2(n)\times L_2(n)$ invertible matrix, and $\B{B}$ and $\B{C}$ are both $L_1(n)\times L_2(n)$ matrices with rank $L_2(n)$. 
$\SR{A}'_1, \SR{A}'_2, \SR{A}'_3$ can be further condensed into a single condition:
\begin{flalign}
\label{eq_v1_align}
\B{T}\B{V}_1\B{C} = \B{V}_1\B{BA}
\end{flalign}
where $\B{T}=\B{M}_{12}\B{M}^{-1}_{21}\B{M}_{23}\B{M}^{-1}_{32}\B{M}_{31}\B{M}^{-1}_{13}$. 
Finally, conditions $\SR{B}_1,  \SR{B}_2, \SR{B}_3$ are reformulated as:
\begin{flalign*}
& \SR{B}'_1: \hspace{8pt} \psi_1(\xi)=\det(\B{V}_1 \quad \B{P}_1\B{V}_1\B{C}) \neq 0 \\
& \SR{B}'_2: \hspace{8pt} \psi_2(\xi)=\det(\B{V}_1 \quad \B{P}_2\B{V}_1\B{C}) \neq 0 \\
& \SR{B}'_3: \hspace{8pt} \psi_3(\xi)=\det(\B{V}_1 \quad \B{P}_3\B{V}_1\B{C}\B{A}^{-1}) \neq 0
\end{flalign*}
where $\B{P}_1=\B{M}_{31}\B{M}^{-1}_{11}\B{M}_{12}\B{M}^{-1}_{32}$, $\B{P}_2=\B{M}_{31}\B{M}^{-1}_{21}\B{M}_{22}\B{M}^{-1}_{32}$, and $\B{P}_3=\B{M}_{12}\B{M}^{-1}_{32}\B{M}_{33}\B{M}^{-1}_{13}$, and $\psi_1(\xi),\psi_2(\xi), \psi_3(\xi)$ are rational functions in the field $\BF_{2^m}(\xi)$. 
Define $\psi(\xi)=\prod^3_{i=1}\psi_i(\xi)$. We assume that $\BF_{2^m}$ is sufficiently large such that if $\psi(\xi)$ is a non-zero rational function, there are values to $\xi$, denoted by $\xi_0$, such that $\psi(\xi_0)\neq 0$.

We also define the following rational functions:
\vspace{-5pt}
\begin{small}
\begin{flalign}
\label{def_pi}
\begin{split}
& p_1(\B{x})=\frac{m_{31}(\B{x})m_{12}(\B{x})}{m_{11}(\B{x})m_{32}(\B{x})} \hspace*{5pt}
p_2(\B{x})=\frac{m_{31}(\B{x})m_{22}(\B{x})}{m_{21}(\B{x})m_{32}(\B{x})} \\
& p_3(\B{x})=\frac{m_{12}(\B{x})m_{33}(\B{x})}{m_{32}(\B{x})m_{13}(\B{x})} \hspace*{5pt}
\eta(\B{x})=\frac{m_{31}(\B{x})m_{12}(\B{x})m_{23}(\B{x})}{m_{21}(\B{x})m_{32}(\B{x})m_{13}(\B{x})}
\end{split}
\end{flalign}
\end{small}\vspace*{-5pt}

Clearly, $p_i(\B{x})$ and $\eta(\B{x})$ form the elements along the diagonals of $\B{P}_i$ and $\B{T}$ respectively. Hence, the following lemma holds:
\begin{lemma}
\label{lemma_na}
$\B{R}^*_n$ is feasible through PBNA if and only if 1) Eq. (\ref{eq_v1_align}) is satisfied, and 2) $\SR{B}'_1,  \SR{B}'_2, \SR{B}'_3$ are satisfied.
\end{lemma}

Form Lemma \ref{lemma_na}, we see that a solution to PBNA consists of four matrices, i.e., $\B{V}_1$, $\B{A}$, $\B{B}$ and $\B{C}$.
We use vector $\B{\Gamma}$ to represent such a solution:
\begin{flalign}
\label{eq_na_solution}
\B{\Gamma} = (\B{V}_1, \B{A}, \B{B}, \B{C})
\end{flalign}

The fundamental design problem in PBNA is to find $\B{\Gamma}$ such that all the conditions in Lemma \ref{lemma_na} are satisfied. 
Indeed, the major restriction comes from Eq. (\ref{eq_v1_align}). 
As shown in \cite{Das2010}, the construction of $\B{\Gamma}$ depends on whether $\eta(\B{x})$ is constant. 
When $\eta(\B{x})$ is constant, and thus $\B{T}$ is an identity matrix, we set $\B{C}=\B{BA}$. 
Therefore, any \textit{arbitrary} $\B{V}_1$ can satisfy Eq. (\ref{eq_v1_align}). 
In fact, for this case, as we will see in Section V-A that all the interferences can be perfectly aligned such that the we can achieve one half rate for each unicast session in exactly two time slots.

In contrast, when $\eta(\B{x})$ is not constant, we can no longer choose $\B{V}_1$ freely. 
\cite{Cadambe2008} proposed the following solution, which has also been used by most of recent work \cite{Das2010}\cite{Ramakrishnan2010}\cite{Han2011}. Let $L_1(n)=n+1$ and $L_2(n)=n$, and define the precoding matrix
\begin{flalign}
\label{eq_v1}
\B{V}^*_1=(\B{w} \hspace*{8pt} \B{T}\B{w} \hspace*{6pt} \cdots \B{T}^n\B{w})
\end{flalign}
where $\B{w}$ is a column vector of $2n+1$ ones. 
Meanwhile, we set $\B{A}=\B{I}_n$, $\B{C}$ consists of the left $n$ columns of $\B{I}_{n+1}$, and $\B{B}$ the right $n$ columns of $\B{I}_{n+1}$; this construction satisfies Eq. (\ref{eq_v1_align}). 
Note that the form of $\B{V}_1$ is determined by $\B{A},\B{B}$ and $\B{C}$. 
With different $\B{A},\B{B}$ and $\B{C}$, we can derive different $\B{V}_1$; therefore the choice of $\B{V}_1$ is not limited to just $\B{V}^*_1$.
Using this solution, we can achieve the following rate tuple through PBNA:
\begin{flalign}
\label{eq_rate_star}
\B{R}^*_n = \bigg( \frac{n+1}{2n+1}, \frac{n}{2n+1}, \frac{n}{2n+1} \bigg)
\end{flalign}

As observed in \cite{Ramakrishnan2010}, graphs can introduce dependence between transfer functions\footnote{Dependence here means that one transfer function (namely $m_{ii}(\B{x})$, corresponding to signal for the $i$th unicast flow) can be written as a rational function of other transfer (interference) functions. The exact functional form is dictated by Eq. (\ref{eq_big_cond}) or Eq. (\ref{eq_small_cond_1})-(\ref{eq_small_cond_3}).}  so that PBNA may be infeasible. 
This is a fundamental difference compared to wireless interference channel, where channel gains can change independently and interference alignment is always feasible. 
Fig. \ref{fig_counter_example}  depicts some examples of graphs where PBNA is infeasible.  
In Fig. \ref{fig_counter_example}(a),  $p_i(\B{x})=\eta(\B{x})=1$ for $i\in \{1,2,3\}$, thus $\B{P}_i=\B{I}_{2n+1}$, implying $\SR{B}'_1,\SR{B}'_2,\SR{B}'_3$ are all violated. 
In Fig. \ref{fig_counter_example}(b), $p_1(\B{x})=\frac{\eta(\B{x})}{\eta(\B{x})+1}$, which also violates $\SR{B}'_1$. 
This example shows that the conjecture proposed by Ramakrishnan et al. \cite{Ramakrishnan2010} doesn't hold beyond three symbol extensions. 
Moreover, by exchanging $s_1\leftrightarrow s_2$ and $d_1\leftrightarrow d_2$, we obtain another counter example, where  $p_2(\B{x})=1+\eta(\B{x})$, violating $\SR{B}'_2$.

\begin{figure}
\centering
\subfloat[$p_1(\B{x})=p_2(\B{x})=p_3(\B{x})=\eta(\B{x})=1$]{\includegraphics[width=4cm]{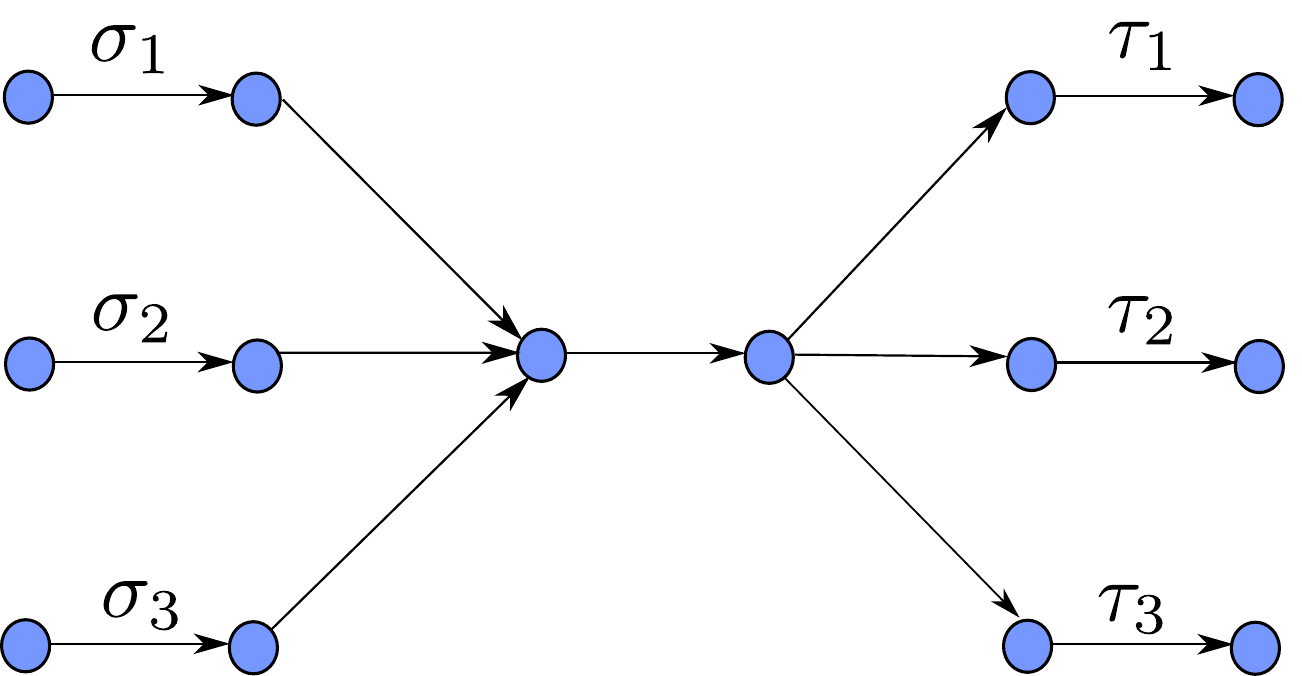}} \hspace*{10pt}
\subfloat[$p_1(\B{x})=\frac{\eta(\B{x})}{\eta(\B{x})+1}$]{\includegraphics[width=4cm]{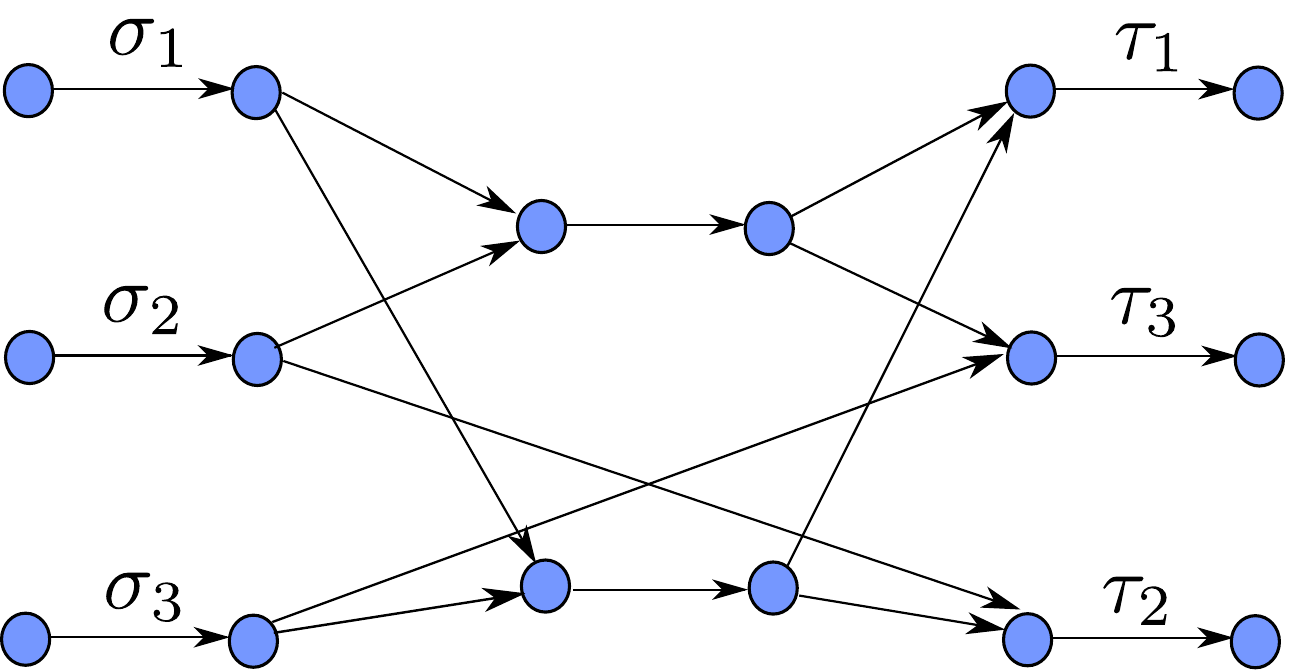}}
\caption{\small{Examples of graphs where PBNA is infeasible}\label{fig_counter_example}} \vspace{-15pt}
\end{figure}

As a first step, \cite{Das2010} proposed the following set of conditions for PBNA.\footnote{There is actually a small difference between Eq. (\ref{eq_big_cond_1}) and the original formulation in \cite{Das2010}, in which $p_1(\B{x})$ is replaced by $1/p_1(\B{x})$. It is easy to see that the two are equivalent.} 
For $i\in \{1,2,3\}$,
\begin{align}
\label{eq_big_cond_1}
p_i(\B{x})\notin\left\{\frac{f(\eta(\B{x}))}{g(\eta(\B{x}))}: f(z),g(z)\in \BF_{2^m}[z], g(z)\neq 0\right\} 
\end{align}
In \cite{Das2010}, it was proved that if Eq. (\ref{eq_big_cond_1}) is satisfied, we can use $\B{V}^*_1$ to asymptotically achieve half rate in an infinite number of time slots. 
Unfortunately, since Eq. (\ref{eq_big_cond_1}) contains an infinite number of conditions, it is impractical to verify. 
Moreover, since only one particular matrix was considered in \cite{Das2010}, Eq. (\ref{eq_big_cond_1}) was only shown to be sufficient for PBNA.

\section{Overview of Main Results\label{sec_main}}
We now state our main results; proofs are deferred to Section V and to the appendices.
Since the construction of $\B{V}_1$ depends on whether $\eta(\B{x})$ is constant, we distinguish  two cases.  

\subsection{$\eta(\B{x})$ Is Constant}

In  this case, we can choose $\B{V}_1$ freely, and thus the feasibility condition of PBNA can be significantly simplified.
Moreover, we can achieve one half rate in exactly two time slots, as stated in the following theorem:
\begin{theorem}
\label{th_na_eta_trivial}
Assume $\eta(\B{x})$ is constant. The rate tuple $(\frac{1}{2}, \frac{1}{2}, \frac{1}{2})$ is feasible through PBNA if and only if $p_i(\B{x})$ is not constant for each $i\in\{1,2,3\}$.
\end{theorem}

\subsection{$\eta(\B{x})$ Is Not Constant}

In this case, we cannot choose $\B{V}_1$ freely. Using similar technique as in \cite{Das2010}, we can rewrite Eq. (\ref{eq_big_cond_1}) as follows: 
\footnote{Notation: For two polynomials $f(x)$ and $g(x)$, let $\mygcd(f(x),g(x))$ denote their greatest common divisor, and $d_f$ the degree of $f(x)$.} 
\begin{flalign}
\begin{split}
\label{eq_big_cond}
p_i(\B{x}) \notin & \C{S}_n = \bigg\{ \frac{f(\eta(\B{x}))}{g(\eta(\B{x}))}: f(z),g(z) \in \BF_{2^m}[z], \\
& f(z)g(z)\neq 0, \gcd(f(z),g(z))=1, \\
& d_f\le n, d_g \le n-1 \bigg\} \hspace*{16pt} \forall i\in \{1,2,3\}
\end{split}
\end{flalign}
Note that, in contrast to Eq. (\ref{eq_big_cond_1}), the above set of conditions guarantee that $\B{R}^*_n$ is NA-feasible for a fixed value of $n$.

Next, we show that Eq. (\ref{eq_big_cond}) is also necessary for the feasibility of PBNA with respect to any $\B{V}_1$ satisfying the conditions of Lemma \ref{lemma_na}.
\begin{theorem}
\label{th_big_cond}
Assume $\eta(\B{x})$ is not constant. $\B{R}^*_n$ is feasible through PBNA if and only if for each $i\in\{1,2,3\}$, $p_i(\B{x})\notin \C{S}_n$.
\end{theorem}

Finally, we greatly reduce $\C{S}_n$ to just four rational functions:
\begin{theorem}[The Main Theorem]
\label{th_main}
Assume $\eta(\B{x})$ is not constant. For $n>1$, $\B{R}^*_n$ is feasible through PBNA if and only if the following conditions are satisfied:
\begin{flalign}
m_{11}(\B{x}) &\neq a_1\frac{m_{21}(\B{x})m_{13}(\B{x})}{m_{23}(\B{x})} + b_1\frac{m_{31}(\B{x})m_{12}(\B{x})}{m_{32}(\B{x})} \label{eq_small_cond_1} \\
m_{22}(\B{x}) &\neq a_2\frac{m_{32}(\B{x})m_{21}(\B{x})}{m_{31}(\B{x})} + b_2\frac{m_{12}(\B{x})m_{23}(\B{x})}{m_{13}(\B{x})} \label{eq_small_cond_2} \\
m_{33}(\B{x}) &\neq a_3\frac{m_{13}(\B{x})m_{32}(\B{x})}{m_{12}(\B{x})} + b_3\frac{m_{23}(\B{x})m_{31}(\B{x})}{m_{21}(\B{x})} \label{eq_small_cond_3}
\end{flalign}
where for $i\in \{1,2,3\}$, $a_i, b_i$ are constants in $\{0, 1\}$ and cannot be zeros at the same time.
\end{theorem}

Note that Eq. (\ref{eq_small_cond_1})-(\ref{eq_small_cond_3}) correspond to the following conditions respectively:
\begin{flalign}
p_1(\B{x}) \notin \bigg\{1, \eta(\B{x}), \frac{\eta(\B{x})}{1+\eta(\B{x})} \bigg\} \label{eq_small_cond_11}\\
p_2(\B{x}) \notin \{1, \eta(\B{x}), 1+\eta(\B{x}) \} \label{eq_small_cond_21} \\
p_3(\B{x}) \notin \{1, \eta(\B{x}), 1+\eta(\B{x}) \} \label{eq_small_cond_31}
\end{flalign}

As shown in the Main Theorem, the feasibility conditions for PBNA are irrelevant to $n$ for $n>1$.
This indicates that if PBNA is feasible for $n=2$, then it is feasible for any arbitrary $n>1$, and thus we can use PBNA to achieve half rate asymptotically.
Otherwise, if PBNA is not feasible for some $n>1$, PBNA doesn't even allow us to achieve any rate greater than $(\frac{2}{3},\frac{1}{3},\frac{1}{3})$.

The basic idea behind the Main Theorem is that we can compare the degree of a variable in $p_i(\B{x})$ with that of a rational function in $\C{S}_n$. This technique enables us to reduce $\C{S}_n$ to the form $\{\frac{a_0+a_1\eta(\B{x})}{b_0+b_1\eta(\B{x})}\}$. Thus, we only need to consider a finite number of rational functions, namely Eq. (\ref{eq_small_cond_1})-(\ref{eq_small_cond_3}). This enables an efficient algorithm for checking the feasibility of PBNA. The key for enabling this reduction lies in two graph-related properties, which we refer to as Linearization Property and Square-Term Property, as described in the next section.

\section{Graph-Related Properties}
Our key intuition is that $p_i(\B{x})$ is not an arbitrary function but depends on transfer functions, as specified in Eq. (\ref{def_pi}). 
Therefore, $p_i(\B{x})$ has special algebraic properties, which can be exploited to simplify Eq. (\ref{eq_big_cond}).

First note that all $p_i(\B{x})$'s are of the following general form: 
\begin{align*}
h(\B{x})=\frac{m_{ab}(\B{x})m_{pq}(\B{x})}{m_{aq}(\B{x})m_{pb}(\B{x})}, \quad  a, b, p, q \in\{1,2,3\},  a \neq p, b \neq q
\end{align*}

Furthermore, each path pair in $\C{P}_{ab}\times\C{P}_{pq}$ contributes a term in  $m_{ab}(\B{x})m_{pq}(\B{x})$, and each path pair in $\C{P}_{aq}\times\C{P}_{pb}$ contributes a term in $m_{aq}(\B{x})m_{pb}(\B{x})$:
\begin{flalign*}
\begin{split}
m_{ab}(\B{x})m_{pq}(\B{x})=\sum_{(P_1,P_2)\in \C{P}_{ab}\times\C{P}_{pq}}\nolimits t(P_1)t(P_2)\\
m_{aq}(\B{x})m_{pb}(\B{x})=\sum_{(P_3,P_4)\in \C{P}_{aq}\times\C{P}_{pb}}\nolimits t(P_3)t(P_4)
\end{split}
\end{flalign*}

\subsection{Linearization Property}

First, consider the following lemma, which provides an easy way to check whether $p_i(\B{x})\notin \{1,\eta(\B{x})\}$ (as in Section VI). 
The intuition is that we can multicast two symbols from $s'_b,s'_q$ to $d'_a,d'_p$ by network coding if and only if the minimum cut separating $s'_b, s'_q$ from $d'_a, d'_p$ is greater than one \cite{Koetter2003}.
\begin{lemma}
\label{lemma_transfer_inequal}
$m_{ab}(\B{x})m_{pq}(\B{x}) \neq m_{aq}(\B{x})m_{pb}(\B{x})$ if and only if there is disjoint path pair $(P_1,P_2)\in \C{P}_{ab}\times \C{P}_{pq}$ or $(P_3,P_4)\in \C{P}_{aq}\times \C{P}_{pb}$.
\end{lemma}
\begin{proof}
See Appendix \ref{app_graph}.
\end{proof}

The first graph-related property states that $p_i(\B{x})$ can be transformed into its simplest non-trivial form (i.e., a linear function or the inverse of a linear function).
The key to Lemma \ref{lemma_nontrivial} is to find a subgraph $H$ and consider $h(\B{x})$ restricted to $H$, i.e., $h(\B{x}_H)=\frac{m_{ab}(\B{x}_H)m_{pq}(\B{x}_H)}{m_{aq}(\B{x_H})m_{pb}(\B{x}_H)}$, where $\B{x}_{H}$ represents the coding vector of $H$. 
Due to the graph structure induced by Lemma \ref{lemma_transfer_inequal}, we can always find $H$ such that some variable $x_{ee'}$ appears exclusively in the numerator or the denominator of $h(\B{x}_H)$. 
Thus, by assigning values to $\B{x}_H$ other than $x_{ee'}$, we can transform $h(\B{x}_H)$ into a linear function or the inverse of a linear function in terms of $x_{ee'}$.
 Since $h(\B{x}_H)$ can be acquired through a partial assignment to $\B{x}$, this transformation also holds for the complete graph $G$.

\begin{figure}[t]
\centering
\subfloat[$o(e_2)>o(e_3)$ and $o(e_1)<o(e_4)$]{\includegraphics[width=2.5cm]{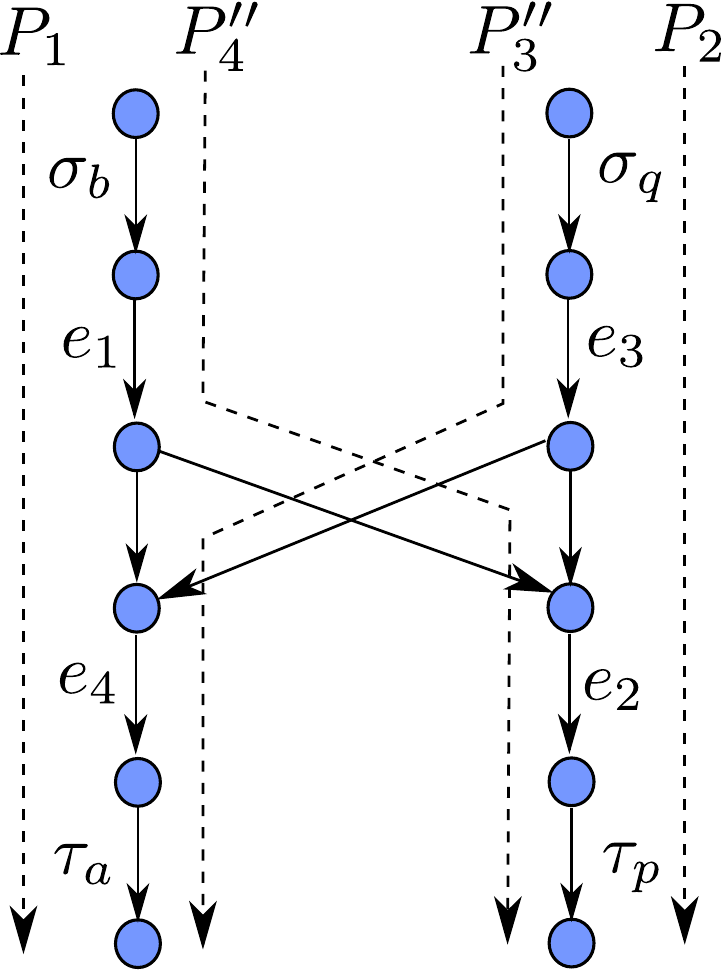}} \hspace*{10pt}
\subfloat[$o(e_2)>o(e_3)$ and $o(e_1)>o(e_4)$]{\includegraphics[width=2.5cm]{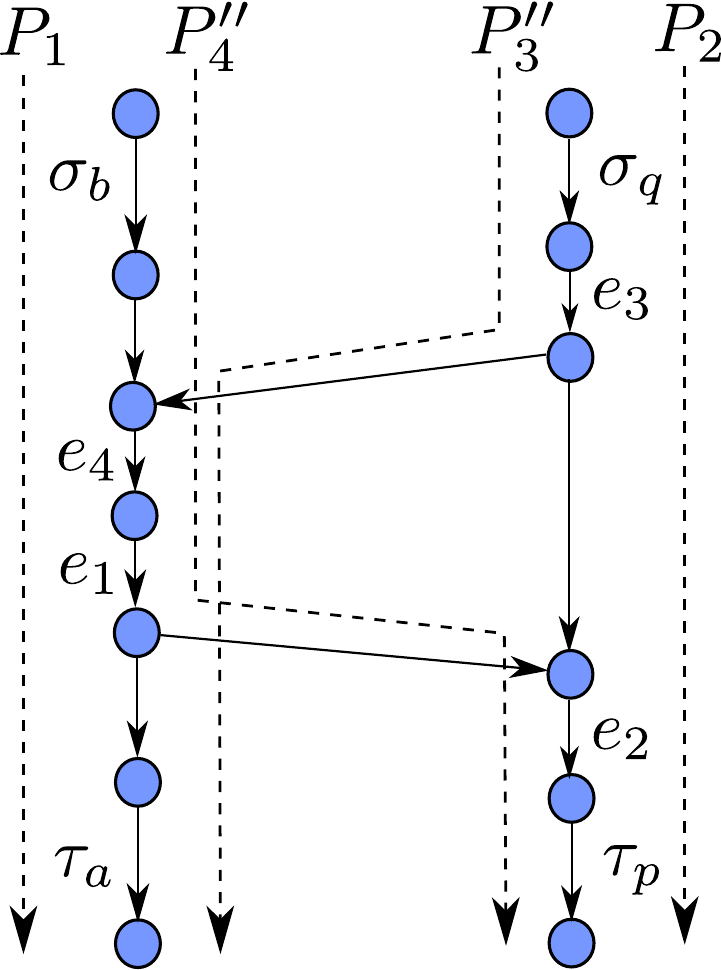}} \hspace*{10pt}
\subfloat[$o(e_2)<o(e_3)$]{\includegraphics[width=2.5cm]{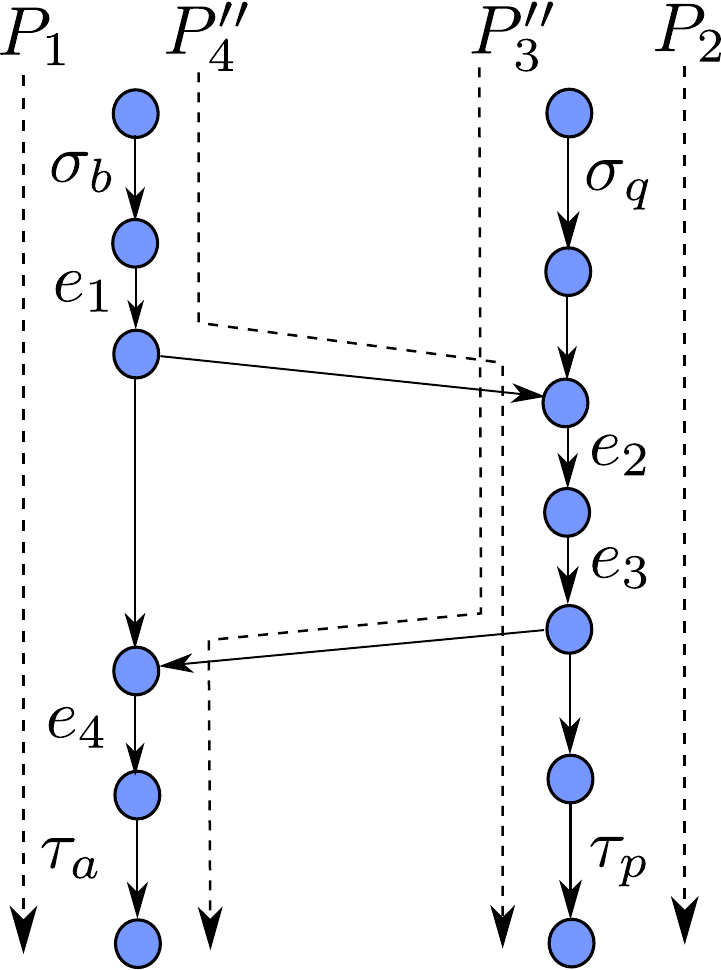}}
\caption{\small{The construction of $H$ (in the proof of the Linearization Property) enabled by Lemma \ref{lemma_transfer_inequal} ($P_1$ is disjoint with $P_2$)} \label{fig_h}}
\end{figure}

\begin{lemma}[Linearization Property]
\label{lemma_nontrivial}
Let $h(\B{x})= \frac{m_{ab}(\B{x})m_{pq}(\B{x})}{m_{aq}(\B{x})m_{pb}(\B{x})}= \frac{u(\B{x})}{v(\B{x})}$ such that $\mygcd(u(\B{x}),v(\B{x}))=1$. Assume $h(\B{x})$ is not constant. Then, for sufficiently large $m$, we can assign values to $\B{x}$ other than a variable $x_{ee'}$ such that $u(\B{x})$ and $v(\B{x})$ are transformed into either $u(x_{ee'}) = c_1x_{ee'} + c_0$, $v(x_{ee'})=c_2$ or $u(x_{ee'}) = c_2, v(x_{ee'})=c_1x_{ee'} + c_0$, where $c_0, c_1, c_2$ are constants in $\BF_{2^m}$, and $c_1c_2 \neq 0$.
\end{lemma}
\begin{proof}
In this proof, given a path $P$ and $e,e'\in P$, let  $P[e:e']$ denote the path segment along $P$ between $e$ and $e'$, including $e,e'$. 
We arrange the edges of $G'$ in topological order, and for $e\in E'$, let $o(e)$ denote $e$'s position in this ordering. 
Moreover, denote $h_1(\B{x})=m_{ab}(\B{x})m_{pq}(\B{x})$, $h_2(\B{x})=m_{aq}(\B{x})m_{pb}(\B{x})$ and $d(\B{x})=\gcd(h_1(\B{x}),h_2(\B{x}))$. 
Let $s_1(\B{x})=\frac{h_1(\B{x})}{d(\B{x})}$ and $s_2(\B{x})=\frac{h_2(\B{x})}{d(\B{x})}$. 
Hence $\gcd(s_1(\B{x}),s_2(\B{x}))=1$. 
It follows $u(\B{x}) = cs_1(\B{x}) \hspace*{15pt} v(\B{x}) = cs_2(\B{x})$, where $c$ is a non-zero constant in $\BF_{2^m}$. 
By Lemma \ref{lemma_transfer_inequal}, there exists disjoint path pair $(P_1,P_2)\in \C{P}_{ab}\times \C{P}_{pq}$ or $(P_3,P_4)\in \C{P}_{aq}\times \C{P}_{pb}$. Now we consider the first case.

Let $(P'_3,P'_4)\in \C{P}_{aq}\times \C{P}_{pb}$. 
Since $P_1,P'_4$ both originate at $\sigma_b$, and $P_2, P'_4$ both terminate at $\tau_p$, there exist $e_1 \in P_1 \cap P'_4$ and $e_2 \in P_2 \cap P'_4$ such that the path segment along $P'_4$ between $e_1$ and $e_2$ is disjoint with $P_1\cup P_2$. 
Similarly, there exist $e_3 \in P_2 \cap P'_3$ and $e_4 \in P_1 \cap P'_3$ such that the path segment between $e_3$ and $e_4$ along $P'_3$ is disjoint with $P_1\cup P_2$. 
Construct the following two paths: $P''_4=P_1[\sigma_b:e_1] \cup P'_4[e_1:e_2] \cup P_2[e_2:\tau_p]$ and $P''_3=P_2[\sigma_q:e_3]\cup P'_3[e_3:e_4]\cup P_1[e_4:\tau_a]$ (see Fig. \ref{fig_h}). 
Let $H$ denote the subgraph of $G'$ induced by $P_1\cup P_2 \cup P''_3 \cup P''_4$, and $\B{x}_H$ the coding vector of $H$. 
We will prove that the theorem holds for $H$. Note that since $h_1(\B{x}_H)$ and $h_2(\B{x}_H)$ are both non-zeros, $d(\B{x}_H) \neq 0$.

If $o(e_2)>o(e_3)$ (Fig. \ref{fig_h}(a)-(b)), the variables in $t(P_2[e_3:e_2])$ are absent in $h_2(\B{x}_H)$.
 We then arbitrarily select a variable $x_{ee'}$ from $t(P_2[e_3:e_2])$, and write $h_1(\B{x}_H)$ as $f(\B{x}'_H)x_{ee'} + g(\B{x}'_H)$, where $\B{x}'_H$ includes all the variables in $\B{x}_H$ other than $x_{ee'}$, and $f(\B{x}'_H), g(\B{x}'_H) \in \BF_{2^m}[\B{x}'_H]$. 
Meanwhile, $h_2(\B{x}_H)$ can be written as $h_2(\B{x}'_H) \in \BF_{2^m}[\B{x}'_H]$. Clearly, $x_{ee'}$ will not show up in $d(\B{x}_H)$ and thus it can also be written as $d(\B{x}'_H)\in \BF_{2^m}[\B{x}'_H]$. 
We then find values for $\B{x}'_H$, denoted by $\B{r}$, such that $f(\B{r})h_2(\B{r})d(\B{r}) \neq 0$. 
Finally, denote $c_0=cg(\B{r})d^{-1}(\B{r})$, $c_1=cf(\B{r})d^{-1}(\B{r})$ and $c_2=ch_2(\B{r})d^{-1}(\B{r})$ and the theorem holds.

On the other hand, if $o(e_2)<o(e_3)$ (see Fig. \ref{fig_h}(c)), the variables in $t(P_1[e_1:e_4])$ are absent in $h_2(\B{x}_H)$. 
We then select a variable $x_{ee'}$ from $t(P_1[e_1:e_4])$. 
Similar to above, it's easy to see that $u(\B{x})$ and $v(\B{x})$ can be transformed into $c_1x_{ee'}+c_0$ and $c_2$ respectively.

For the case where there exists disjoint path pair $(P_3,P_4)\in \C{P}_{aq}\times\C{P}_{pb}$, we can show that $u(\B{x})$ and $v(\B{x})$ can be transformed into $c_2$ and $c_1x_{ee'}+c_0$ respectively.
\end{proof}

\subsection{Square-Term Property}
\begin{figure}[t]
\centering
\includegraphics[width=4cm]{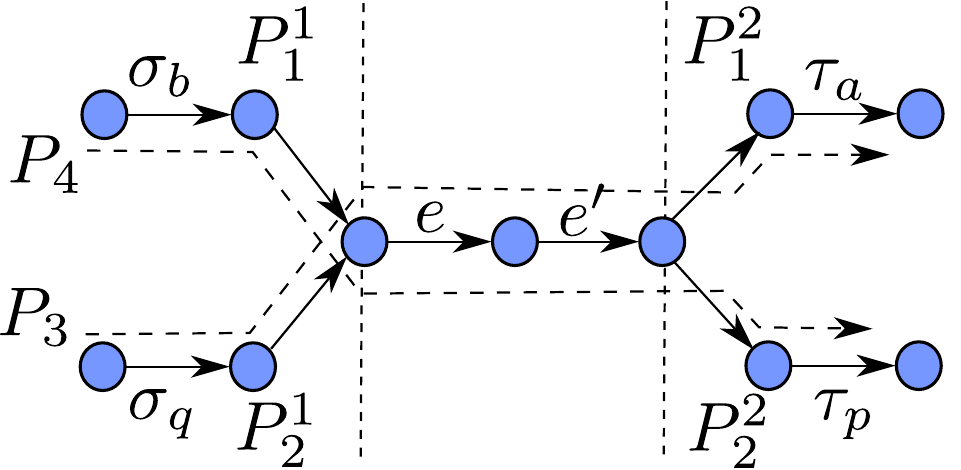}
\caption{\small{Illustration of Square-Term Property. A term with $x^2_{ee'}$ introduced by $(P_1,P_2)$ in the numerator of $h(\B{x})$ equals another term introduced by $(P_3,P_4)$ in the denominator of $h(\B{x})$.}} \label{fig_square_term}\vspace*{-0.6cm}
\end{figure}

The second graph-related property is stated in Lemma \ref{lemma_x2}: the coefficient of $x^2_{ee'}$ in the numerator of $h(\B{x})$ equals its counter-part in the denominator of $h(\B{x})$. 
Thus, if $x^2_{ee'}$ appears in the numerator of $h(\B{x})$ under some assignment to $\B{x}$, it must also appear in the denominator of $h(\B{x})$, and vice versa.
\begin{lemma}[Square-Term Property]
\label{lemma_x2}
Given a coding variable $x_{ee'}$, let $f_1(\B{x})$ and $f_2(\B{x})$ be the coefficients of $x^2_{ee'}$ in $m_{ab}(\B{x})m_{pq}(\B{x})$ and $m_{aq}(\B{x})m_{pb}(\B{x})$ respectively. Then $f_1(\B{x})=f_2(\B{x})$.
\end{lemma}
\begin{proof}
For any $x_{ee'}$, define $\C{Q}_1 = \{(P_1,P_2)\in \C{P}_{ab}\times \C{P}_{pq}: x^2_{ee'} \mid t(P_1)t(P_2)\}$ and $\C{Q}_2 = \{(P_3,P_4)\in \C{P}_{aq}\times \C{P}_{pb}: x^2_{ee'} \mid t(P_3)t(P_4)\}$. Consider a path pair $(P_1,P_2)\in \C{Q}_1$. Since the degree of $x_{ee'}$ in $t(P_1)$ and $t(P_2)$ is at most one, we must have $x_{ee'} \mid t(P_1)$ and $x_{ee'} \mid t(P_2)$. Thus $e,e'\in P_1\cap P_2$. Let $P^1_1, P^2_1$ be the parts of $P_1$ before $e$ and after $e'$ respectively. Similarly, define $P^1_2$ and $P^2_2$. Then construct two new paths: $P_3=P^1_2 \cup \{e,e'\} \cup P^2_1$ and $P_4=P^1_1\cup \{e,e'\} \cup P^2_2$ (see Fig. \ref{fig_square_term}). Clearly, $t(P_1)t(P_2)=t(P_3)t(P_4)$, and thus $(P_3,P_4) \in \C{Q}_2$. The above method establishes a one-to-one mapping $\phi:\C{Q}_1\rightarrow\C{Q}_2$, such that for $\phi((P_1, P_2))=(P_3,P_4)$, $t(P_1)t(P_2)=t(P_3)t(P_4)$. Hence, $f_1(\B{x})=\frac{1}{x^2_{ee'}}\sum_{(P_1,P_2)\in\C{Q}_1}t(P_1)t(P_2)=\frac{1}{x^2_{ee'}}\sum_{(P_3,P_4)\in \C{Q}_2}t(P_3)t(P_4)=f_2(\B{x})$.
\end{proof}

\section{Feasibility Condition of \NA \label{sec_cond_na}}
In this section, we provide the proofs of Theorems \ref{th_na_eta_trivial}, \ref{th_big_cond} and 3 (Main Theorem).

\subsection{$\eta(\B{x})$ Is Constant}
\begin{proof}[Proof of Theorem \ref{th_na_eta_trivial}]
In this case, $\B{T}$ is identity matrix. We set $L_1(n)=L_2(n)=1$ and $\B{V}_1=(\theta_1 \hspace*{5pt} \theta_2)^T$, where $\theta_1, \theta_2$ are arbitrary variables, and $\B{A}, \B{B}, \B{C}$ are all scalar  ones. It is easy to see that Eq. (\ref{eq_v1_align}) is satisfied. Moreover, if $p_i(\B{x})$ is not constant, we have 
\begin{flalign*}
\psi_i(\xi) = \det\begin{pmatrix}
\theta_1 & p_i(\B{x}^1)\theta_1 \\
\theta_2 & p_i(\B{x}^2)\theta_2
\end{pmatrix} = \theta_1\theta_2(p_i(\B{x}^1)-p_i(\B{x}^2)) \neq 0
\end{flalign*}
and $\SR{B}'_i$ is satisfied. Thus $(\frac{1}{2}, \frac{1}{2}, \frac{1}{2})$ is feasible through PBNA. Conversely, if $p_i(\B{x})$ is constant, $\SR{B}'_i$ is violated, and thus $(\frac{1}{2}, \frac{1}{2}, \frac{1}{2})$ is not feasible through PBNA.
\end{proof}

\subsection{$\eta(\B{x})$ Is Not Constant}
Due to the importance of $\B{V}_1$, we first consider how to construct $\B{V}_1$ which satisfies (\ref{eq_v1_align}). 
The construction of $\B{V}_1$ involves solving a system of linear equations:
\begin{flalign}
\label{eq_v1_align_2}
\B{r}(z)(z\B{C}-\B{BA})=0
\end{flalign}
where $\B{r}(z)=(r_1(z),\cdots,r_{n+1}(z))\in \BF^{n+1}_{2^m}(z)$. 
It is easy to see that $z\B{C}-\B{BA}$ is a matrix on $\BF_{2^m}(z)$. 
Assume $\B{r}_0(z)$ is a non-zero solution to (\ref{eq_v1_align_2}). 
Substitute $z$ with $\eta(\B{x}^i)$, and we have
\begin{flalign*}
\eta(\B{x}^i)\B{r}_0(\eta(\B{x}^i))\B{C} = \B{r}_0(\eta(\B{x}^i))\B{BA}
\end{flalign*}
Finally, construct the following precoding matrix
\begin{flalign*}
\B{V}^T_1=(\B{r}^T_0(\eta(\B{x}^1)) \hspace*{6pt} \B{r}^T_0(\eta(\B{x}^2)) \hspace*{6pt} \cdots \hspace*{6pt} \B{r}^T_0(\eta(\B{x}^{2n+1})))
\end{flalign*}
Apparently, $\B{V}_1$ satisfies (\ref{eq_v1_align}). 
Hence, each non-zero solution to (\ref{eq_v1_align_2}) corresponds to a row of $\B{V}_1$ satisfying (\ref{eq_v1_align}). 
Conversely, it is straightforward to see that each row of $\B{V}_1$ satisfying (\ref{eq_v1_align}) corresponds to a solution to (\ref{eq_v1_align_2}).

\begin{example}
\label{ex_v1_not_unique}
As an example, consider the case where $n=2$ and $m=2$. 
Let $\alpha$ be the primitive element of $\BF_4$ such that $\alpha^3=1$ and $\alpha^2+\alpha+1=0$. Moreover, let $\B{A}=\B{I}_2$ and
\begin{flalign*}
\B{C}=\begin{pmatrix}
1 & \alpha \\
\alpha & 1 \\
\alpha^2 & 1
\end{pmatrix} \hspace*{10pt}
\B{B}=\begin{pmatrix}
\alpha^2 & \alpha \\
1 & 1 \\
1 & \alpha
\end{pmatrix}
\end{flalign*}
Apparently, $\myrank(\B{C})=\myrank(\B{B})=2$. 
It's easy to verify that $\B{r}(z)=(\alpha^2z^2+\alpha, z+\alpha, z^2+\alpha z+\alpha^2)$ satisfies equation (\ref{eq_v1_align_2}). 
Thus, we substitute $z$ with $\eta(\B{x}^j)$ and construct $\B{V}^T_1=(\B{r}^T(\eta(\B{x}^1)) \hspace*{8pt} \B{r}^T(\eta(\B{x}^2)) \hspace*{6pt}\cdots\hspace*{6pt} \B{r}^T(\eta(\B{x}^5)))$. 
According to the above discussion, equation (\ref{eq_v1_align}) is satisfied. \myqed
\end{example}

Using (\ref{eq_v1_align_2}), we can derive the general form of $\B{V}_1$ which satisfies $\B{V}_1$.
\begin{lemma}
\label{lemma_v1}
Any $\B{V}_1$ satisfying (\ref{eq_v1_align}) has the form $\B{V}_1=\B{G}\B{V}^*_1\B{F}$, where $\B{V}^*_1$ is defined in (\ref{eq_v1}), $\B{F}$ is an $(n+1)\times(n+1)$ matrix, and $\B{G}$ is a $(2n+1)\times (2n+1)$ diagonal matrix, with the $(i,i)$ element being $f_i(\eta(\B{x}^i))$, where $f_i(z)$ is a non-zero rational function in $\BF_{2^m}(z)$. Moreover, the $(n+1)$th row of $\B{FC}$ and the 1st row of $\B{FBA}$ are both zero vectors.
\end{lemma}
\begin{proof}
See Appendix \ref{app_v1}.
\end{proof}

Lemma \ref{lemma_v1} indicates that there is a direct relation between $\B{V}^*_1$ and the general form of $\B{V}_1$, which we use to prove that Eq. (\ref{eq_big_cond}) is also necessary for the feasibility of PBNA.

\begin{proof}[Proof of Theorem \ref{th_big_cond}]
The sufficiency of (\ref{eq_big_cond}) was proved in \cite{Das2010}. 
Now assume $p_i(\B{x})=\frac{f(\eta(\B{x}))}{g(\eta(\B{x}))} \in \C{S}_n$, where $f(z)=\sum^n_{k=0}a_kz^k$ and $g(z)=\sum^{n-1}_{k=0}b_kz^k$. 
We will prove that for any $\B{V}_1$ satisfying (\ref{eq_v1_align}), $\SR{B}'_i$ cannot be satisfied, thus $\B{R}^*_n$ is not NA-feasible. 
Apparently, if $\myrank(\B{V}_1)< n+1$, $\SR{B}'_i$ is violated. Thus, in the rest of this proof, we assume $\myrank(\B{V}_1)=n+1$. 

By Lemma \ref{lemma_v1}, $\B{V}_1=\B{G}\B{V}^*_1\B{F}$, where $\B{F}$ is an $(n+1)\times (n+1)$ invertible matrix. 
The $j$th row of $\B{V}_1$ is $\B{r}_j=f_j(\eta(\B{x}^j))(1, \eta(\B{x}^j), \cdots, \eta^n(\B{x}^j))\B{F}$. 
Since the $(n+1)$th row of $\B{FC}$ is zero, we have 
\begin{align*}
\B{r}_j\B{C} = f_j(\eta(\B{x}^j))(1, \eta(\B{x}^j), \cdots, \eta^{n-1} (\B{x}^j))\B{H}
\end{align*}
where $\B{H}$ consists of the top $n$ rows of $\B{FC}$ and $\myrank(\B{H})=n$. 
Let $\B{a}=(a_0,a_1,\cdots,a_n)^T$ and $\B{b}=(b_0,b_1,\cdots,b_{n-1})^T$. 
For $i=1,2$, we define $\B{a}'=\B{F}^{-1}\B{a}$ and $\B{b}'=\B{H}^{-1}\B{b}$. 
It follows
\begin{flalign*}
\B{r}_j\B{a}' &= f_j(\eta(\B{x}^j))(1, \eta(\B{x}^j), \cdots, \eta^n(\B{x}^j))\B{F}\B{a}' \\
&= f_j(\eta(\B{x}^j))(1, \eta(\B{x}^j), \cdots, \eta^n(\B{x}^j))\B{a} \\
&= f_j(\eta(\B{x}^j))f(\eta(\B{x}^j)) \\
&= f_j(\eta(\B{x}^j))p_i(\B{x}^j)g(\eta(\B{x}^j)) \\
&= p_i(\B{x}^j)f_j(\eta(\B{x}^j))(1, \eta(\B{x}^j), \cdots, \eta^{n-1}(\B{x}^j))\B{b} \\
&= p_i(\B{x}^j)f_j(\eta(\B{x}^j))(1, \eta(\B{x}^j), \cdots, \eta^{n-1}(\B{x}^j))\B{Hb}' \\
&= p_i(\B{x}^j)\B{r}_j\B{Cb}'
\end{flalign*}
Hence, the columns of $(\B{V}_1 \hspace*{6pt} \B{P}_i\B{V}_1\B{C})$ are linearly dependent, violating $\SR{B}'_i$. 
Similarly, we can prove the case of $i=3$.
\end{proof}

For the proof of the Main Theorem, we need to rearrange the ratio of rational functions $\frac{f(\eta(\B{x}))}{g(\eta(\B{x}))}$ in Eq. (\ref{eq_big_cond}) to a ratio of coprime polynomials with variables $\B{x}$. To this end, we use a property of polynomials stated in the following lemma. 
\begin{lemma}
\label{lemma_relative_prime_multivar}
Let $\BF$ be a field. $z$ is a variable and $\B{y}=(y_1,y_2,\cdots,y_k)$ is a vector of variables. 
Consider four non-zero polynomials $f(z),g(z)\in \BF[z]$  and $s(\B{y}), t(\B{y})\in \BF[\B{y}]$, such that $\mygcd(f(z),g(z))=1$ and $\mygcd(s(\B{y}), t(\B{y}))=1$. 
Denote $d=\max\{d_f,d_g\}$. 
Define two polynomials in $\BF[\B{y}]$: $\alpha(\B{y})=f(\frac{s(\B{y})}{t(\B{y})})t^d(\B{y})$ and $\beta(\B{y})=g(\frac{s(\B{y})}{t(\B{y})})t^d(\B{y})$. 
Then $\mygcd(\alpha(\B{y}),\beta(\B{y}))=1$.
\end{lemma}
\begin{proof}
See Appendix \ref{app_polynomial}.
\end{proof}

The proof of the Main Theorem consists of three steps. 
In the first step, we use degree-counting technique and Linearization Property to reduce $\C{S}_n$ to the form $\{\frac{a_0+a_1\eta(\B{x})}{b_0+b_1\eta(\B{x})}\}$. 
In the second step, we use Linearization Property and Square Term Property to further reduce $\C{S}_n$ to the four rational functions in $\C{S}'=\{1, \eta(\B{x}), 1+\eta(\B{x}), \frac{\eta(\B{x})}{1+\eta(\B{x})}\}$.
Finally, we use the results from \cite{Han2011} to rule out the remaining redundant conditions.

\begin{proof}[Proof of the Main Theorem]
Clearly, the necessity of (\ref{eq_small_cond_1})-(\ref{eq_small_cond_3}) (or Eq. (\ref{eq_small_cond_11})-(\ref{eq_small_cond_31})) follows directly from Theorem \ref{th_big_cond}. 
Now assume for $i\in \{1,2,3\}$,  $p_i(\B{x}) \notin \C{S}'$. 
We will prove that $p_i(\B{x}) \notin \C{S}_n$ and thus $\B{R}^*_n$ is NA-feasible by Theorem \ref{th_big_cond}. 
We only prove $p_1(\B{x}) \notin \C{S}_n$. The other cases follow similar lines. 
By contradiction, assume there exists $p_1(\B{x})=\frac{f(\eta(\B{x}))}{g(\eta(\B{x}))}\in \C{S}_n$, where $f(z) = \sum^k_{i=0} a_iz^i$ and $g(z) = \sum^l_{i=0} b_iz^i$ such that $a_lb_k \neq 0$ and $\gcd(f(z),g(z))=1$. 
Moreover, let $p_1(\B{x})=\frac{u(\B{x})}{v(\B{x})}$ and $\eta(\B{x})=\frac{s(\B{x})}{t(\B{x})}$, where $\gcd(u(\B{x}),v(\B{x}))=\gcd(s(\B{x}),t(\B{x}))=1$. 
Let $d=\max\{k,l\}$. 
Define the following two polynomials $\alpha(\B{x})=f(\eta(\B{x}))t^d(\B{x})$ and $\beta(\B{x})=g(\eta(\B{x}))t^d(\B{x})$. 
According to Lemma \ref{lemma_relative_prime_multivar}, $\gcd(\alpha(\B{x}),\beta(\B{x}))=1$. 
Thus, we have $\alpha(\B{x}) = cu(\B{x})$, and $\beta(\B{x}) = cv(\B{x})$, where $c\in \BF_{2^m}$ and $c\neq 0$.

According to Lemma \ref{lemma_nontrivial}, there exists an assignment to $\B{x}$ under which $u(\B{x})$ and $v(\B{x})$ are transformed into either $u(x_{ee'})=c_1x_{ee'}+c_0$, $v(x_{ee'})=c_2$ or $u(x_{ee'})=c_2$, $v(x_{ee'})=c_1x_{ee'}+c_0$. 
We only consider the first case. The proof for the other case is similar. 
In this case, $\alpha(\B{x})$ and $\beta(\B{x})$ are transformed into $\alpha(x_{ee'}) = cc_1x_{ee'}+cc_0$ and $\beta(x_{ee'}) = cc_2$ respectively.

First, we prove that both $t(x_{ee'})$ and $s(x_{ee'})$ are non-zeros. Assume $t(x_{ee'})=0$. 
If $k\neq l$, at least one of $\alpha(x_{ee'})$ and $\beta(x_{ee'})$ equals zero, which is impossible. 
On the other hand, if $k=l$, we have $\alpha(x_{ee'})=a_ks^k(x_{ee'})$ and $\beta(x_{ee'})=b_ks^k(x_{ee'})$. 
It follows that $cc_1x_{ee'}+cc_0 = a_kb^{-1}_kcc_2$, which is impossible. 
Thus we have proved that $t(x_{ee'})\neq 0$. 
Similarly, we can also prove that $s(x_{ee'})\neq 0$. 

We then prove that $d=1$. 
By contradiction, assume $d\ge 2$. 
We first consider the case where $l\le k$ and thus $d=k$. 
In this case, we have 
\begin{align*}
&\alpha(x_{ee'}) = \sum^k_{j=0}\nolimits a_jt^{k-j}(x_{ee'})s^j(x_{ee'}) = cc_1x_{ee'}+cc_0 \\
&\beta(x_{ee'}) = \sum^l_{j=0}\nolimits b_jt^{k-j}(x_{ee'})s^j(x_{ee'}) = cc_2
\end{align*}
Assume $s(x_{ee'})=\sum^r_{j=0}s_jx^j_{ee'}$ and $t(x_{ee'})=\sum^{r'}_{j=0}t_jx^j_{ee'}$, where $s_rt_{r'} \neq 0$. 
Thus $r=d_s$ and $r'=d_t$ and $\max\{r,r'\}\ge 1$. 
Note that the degree of $x_{ee'}$ in $t^{k-j}(x_{ee'})s^j(x_{ee'})$ is $kr'+j(r-r')$. 
We consider the following two cases:

Case I: $r \neq r'$. 
If $r > r'$, $d_{\alpha}=kr \ge 2$, contradicting that $d_{\alpha}=1$. 
Now assume $r<r'$. Let $l_1$ and $l_2$ be the minimum exponents of $z$ in $f(z)$ and $g(z)$ respectively. 
It follows that $d_{\alpha} = kr'-l_1(r'-r) = 1$ and $d_{\beta} = kr'-l_2(r'-r) = 0$. 
Clearly, $l_2>0$ due to $d_{\beta}=0$. 
If $r>0$, $kr'-l_2(r'-r) > kr'-l_2r' \ge 0$, contradicting $d_{\beta} = 0$. 
Hence, $r=0$, and $l_2=k$ due to $d_{\beta}=0$. 
Meanwhile, $d_{\alpha} = (k-l_1)r' = 1$, which implies that $l_1=k-1$ and $r'=1$. 
Thus, $z^{k-1}$ is a common divisor of $f(z)$ and $g(z)$, contradicting $\gcd(f(z),g(z))=1$.

Case II: $r=r'$. Since $d_{\alpha}=1$ and $d_{\beta}(x_{ee'})=0$, all the terms in $\alpha(x_{ee'})$ and $\beta(x_{ee'})$ containing $x^{kr}_{ee'}$ must be cancelled out, implying that 
\begin{align*}
& \sum^k_{j=0} a_jt^{k-j}_rs^j_r = t^k_r \sum^k_{j=0} a_j\left(\frac{s_r}{t_r}\right)^j = t^k_r f\left(\frac{s_r}{t_r}\right) = 0 \\
& \sum^l_{j=0} b_jt^{k-j}_rs^j_r = t^k_r \sum^l_{j=0} b_j\left(\frac{s_r}{t_r}\right)^j = t^k_r g\left(\frac{s_r}{t_r}\right) = 0
\end{align*}
Hence $z-\frac{s_r}{t_r}$ is a common divisor of $f(z)$ and $g(z)$, contradicting $\gcd(f(z),g(z))=1$. 

Therefore, we have proved $d=1$ when $l\le k$. Using similar technique, we can prove that $d=1$ when $l\ge k$. 

Define $q_1(\B{x})=\frac{\eta(\B{x})}{p_1(\B{x})}=\frac{m_{11}(\B{x})m_{23}(\B{x})}{m_{13}(\B{x})m_{21}(\B{x})}$. 
For $d=1$, we consider the following cases.

Case I: $\frac{f(z)}{g(z)}=\frac{a_0+a_1z}{b_0+b_1z}$, where $a_1a_0b_1b_0 \neq 0$, and $a_0b_1 \neq a_1b_0$. 
For this case, we have $p_1(x_{ee'}) = \frac{a_0+a_1p_1(x_{ee'})q_1(x_{ee'})}{b_0+b_1p_1(x_{ee'})q_1(x_{ee'})}$. It immediately follows
\begin{flalign*}
q_1(x_{ee'}) = \frac{a_0c^2_2-b_0c_0c_2-b_0c_1c_2x_{ee'}}{b_1c^2_1x^2_{ee'}-a_1c_1c_2x_{ee'}+b_1c^2_0-a_1c_0c_2}
\end{flalign*}
Denote $u_1(x_{ee'})=a_0c^2_2-b_0c_0c_2-b_0c_1c_2x_{ee'}$ and $v_1(x_{ee'})=b_1c^2_1x^2_{ee'}-a_1c_1c_2x_{ee'}+b_1c^2_0-a_1c_0c_2$. 
Assume $u_1(x_{ee'}) \mid v_1(x_{ee'})$ and thus $x_{ee'}=\frac{a_0c_2-b_0c_0}{b_0c_1}$ is a solution to $v_1(x_{ee'})=0$. 
However, $v_1(\frac{a_0c_2-b_0c_0}{b_0c_1}) = \frac{a_0c^2_2}{b^2_0}(a_0b_1 + a_1b_0) \neq 0$. 
Hence, $u_1(x_{ee'}) \nmid v_1(x_{ee'})$. Thus, by the definition of $q_1(\B{x})$ and Lemma \ref{lemma_x2}, $x^2_{ee'}$ must appear in $u_1(x_{ee'})$, which contradicts the formulation of $u_1(x_{ee'})$.

Case II: $\frac{f(z)}{g(z)}=\frac{a_0+a_1z}{b_1z}$, where $a_0a_1b_0 \neq 0$. 
Similar to Case I, we can derive 
\begin{align*}
q_1(x_{ee'}) = \frac{a_0c^2_2}{b_1c^2_1x^2_{ee'}-a_1c_1c_2x_{ee'}+b_1c^2_0-a_1c_0c_2}
\end{align*}
which contradicts Lemma \ref{lemma_x2}.

Case III: $\frac{f(z)}{g(z)}=\frac{a_1z}{b_0+b_1z}$, where $a_1b_0b_1 \neq 0$. Thus 
\begin{align*}
q_1(\B{x})=\frac{a_1}{b_1} - \frac{b_0}{b_1}\frac{m_{13}(\B{x})m_{21}(\B{x})}{m_{11}(\B{x})m_{23}(\B{x})}
\end{align*}
Since the coefficient of each monomial in $m_{11}(\B{x})m_{23}(\B{x})$ and $m_{13}(\B{x})m_{21}(\B{x})$ equals one, it directly follows $\frac{a_1}{b_1}=-\frac{b_0}{b_1}=\frac{b_0}{b_1}=1$. 
This indicates that $p_1(\B{x})=\frac{\eta(\B{x})}{\eta(\B{x})+1}$, contradicting $p_1(\B{x}) \notin \C{S}'$.

Case IV: $\frac{f(z)}{g(z)}=\frac{a_0}{b_0+b_1z}$, where $a_0b_0b_1\neq 0$. 
It follows that 
\begin{align*}
q_1(x_{ee'}) = \frac{a_0c^2_2-b_0c_0c_2-b_0c_1c_2x_{ee'}}{b_1c^2_0+b_1c^2_1x^2_{ee'}}
\end{align*}
Similar to Case I, this also contradicts Lemma \ref{lemma_x2}.

Case V: $\frac{f(z)}{g(z)}=\frac{a_0}{z}$, where $a_0\neq 0$. 
Hence, $q_1(x_{ee'})=\frac{a_0c^2_2}{c^2_1x^2_{ee'}+c^2_0}$, contradicting Lemma \ref{lemma_x2}.

Case VI: $\frac{f(z)}{g(z)}=a_0+a_1z$, where $a_0a_1 \neq 0$. 
Thus, it follows 
\begin{align*}
p_1(\B{x})=a_0 + a_1\frac{m_{31}(\B{x})m_{12}(\B{x})m_{23}(\B{x})}{m_{21}(\B{x})m_{32}(\B{x})m_{13}(\B{x})}
\end{align*}
Similar to Case III, $a_1=a_0=1$, contradicting $p_1(\B{x}) \notin \C{S}'$.

Case VII: $\frac{f(z)}{g(z)}=a_1z$, where $a_1\neq 0$. 
Similar to Case III, $p_1(\B{x})=\eta(\B{x})$, contradicting $p_1(\B{x}) \notin \C{S}'$.

Thus, we have proved that if $p_i(\B{x})\notin\C{S}'$, $p_i(\B{x})\notin \C{S}_n$ and hence $\B{R}^*_n$ is NA-feasible by Theorem \ref{th_big_cond}.
We note that in \cite{Han2011} the authors proved that $p_1(\B{x})\neq 1+\eta(\B{x})$, $p_2(\B{x})\neq \frac{\eta(\B{x})}{1+\eta(\B{x})}$ and $p_3(\B{x})\neq \frac{\eta(\B{x})}{1+\eta(\B{x})}$. Combined with the above results, we have proved that if Eq. (\ref{eq_small_cond_1})-(\ref{eq_small_cond_3}) are satisfied, $\B{R}^*_n$ is feasible through PBNA.
\end{proof}

\subsection{Some $m_{ij}(\B{x})=0$ $(i\neq j)$}
In this case, since the number of interference terms is reduced, at least one of $\SR{A}_1,\SR{A}_2,\SR{A}_1$ is removed, and thus the restriction on $\B{V}_1$ imposed by Eq. (\ref{eq_v1_align}) vanishes. 
Therefore, we can choose $\B{V}_1$ freely, and the feasibility condition of PBNA is greatly simplified.
For example, assume $m_{23}(\B{x})=0$ and all other transfer functions are non-zeros.
Hence $\SR{A}_2$ is removed.
Meanwhile, $\SR{B}'_1,\SR{B}'_2,\SR{B}'_3$ remain the same.
Similarly to Theorem \ref{th_na_eta_trivial}, we can set $\B{V}_1=(\theta_{ij})_{(2n+1)\times (n+1)}$, where $\theta_{ij}$ is an arbitrary variable.
It is easy to see that $\B{R}^*_n$ is feasible through PBNA if and only if $p_i(\B{x})$ is not constant for every $i\in \{1,2,3\}$.
Using similar arguments, we can discuss other cases.

\section{Checking the Feasibility of PBNA}
For a given graph, checking the feasibility of PBNA is now reduced to checking whether Eq. (\ref{eq_small_cond_11})-(\ref{eq_small_cond_31}).  This is a multivariate polynomial identity testing problem. To check whether $p_i(\B{x})\neq 1$, we use Ford-Fulkerson Algorithm, as per  Lemma \ref{lemma_transfer_inequal}. To check whether $p_i(\B{x})\neq \eta(\B{x})$, we define $q_i(\B{x})=\frac{\eta(\B{x})}{p_i(\B{x})}$ and consider $q_i(\B{x})\neq 1$. Therefore, Ford-Fulkerson Algorithm can be used to check this condition as well. For the other conditions ($p_1(\B{x})\neq \frac{\eta(\B{x})}{1+\eta(\B{x})}$ and $p_2(\B{x}), p_3(\B{x}) \neq 1 + \eta(\B{x})$),  it is still not clear what is their interpretation in terms of graph structure. A counter example is shown in Fig. \ref{fig_counter_example}(b). Nevertheless, we can still check the conditions by evaluating the rational functions through $T$ random tests:

\vspace*{5pt}\hrule
\begin{small}
\begin{algorithmic}
\FOR{$k=1$ to $T$}
\STATE Assign random values to $\B{x}$, denoted by $\B{x}_0$
\STATE If $p_1(\B{x}_0)\neq \frac{\eta(\B{x}_0)}{1+\eta(\B{x}_0)}$, return success
\ENDFOR
\STATE Return failure (i.e., $\SR{B}'_i$ is violated)
\end{algorithmic}
\end{small}
\hrule
\vspace*{5pt}

Let $L$ denote the maximum distance from any sender to any receiver in the network. 
Using Lemma 4 of \cite{TraceyHo2006}, we can upper-bound the probability of error as follows.
We consider the case of $i=1$. Other cases follows along similar lines.
Note that Eq. (\ref{eq_small_cond_1}) is equivalent to the following equation:
\begin{align*}
f(\B{x}) &= m_{11}(\B{x})m_{32}(\B{x})m_{23}(\B{x}) + m_{21}(\B{x})m_{32}(\B{x})m_{12}(\B{x}) + \\
&\hspace*{5pt} m_{31}(\B{x})m_{12}(\B{x})m_{23}(\B{x}) =0
\end{align*}
Since the maximum degree of any variable $x_{ee'}$ in a transfer function is at most one, the total degree of each term in $f(\B{x})$ is at most $3L$.
For each random test, the probability of error in checking if Eq. (\ref{eq_small_cond_1}), denoted by $\delta_1$, can be upper bounded by using Lemma 4 of \cite{TraceyHo2006}: $\delta_1 = Pr(f(\B{x}_0)=0\mid f(\B{x})\neq 0) \le 1-\left(1-\frac{3}{2^m}\right)^L$.
Hence, the total probability of error in checking if $p_1(\B{x})\neq 1+\eta(\B{x})$ is $P_1(Error) = \delta^T_1 \le [1-(1-\frac{3}{2^m})^L]^T$. 
Thus, the error can be made arbitrarily small for sufficiently large $m$ and $T$. 
The running time of the algorithm is $O(T|E|D_{in})$, where $D_{in}$ is the maximum in-degree of any node in the network.

\section{Conclusion}
In this paper, we study the feasibility of PBNA for three unicast sessions. We first prove that the set of conditions proposed by \cite{Das2010} are also necessary for the feasibility of PBNA with respect to any valid precoding matrix. Then, we reduce this set of conditions to just four conditions, using two graph-related properties along with a simple degree-counting technique. This reduction enables an efficient algorithm for checking the feasibility of PBNA.

\appendices
\section{Proofs of Graph Properties \label{app_graph}}
The following lemma is used in the proof of Lemma \ref{lemma_transfer_inequal}.
\begin{lemma}
\label{lemma_equal_term}
Let $(P_1,P_2)\in \C{P}_{ab}\times \C{P}_{pq}$. Then, there exists $(P_3,P_4)\in \C{P}_{aq}\times\C{P}_{pb}$ such that $t(P_1)t(P_2)=t(P_3)t(P_4)$ if and only if $P_1 \cap P_2 \neq \emptyset$.
\end{lemma}
\begin{proof}
First, Assume $P_1 \cap P_2 \neq \emptyset$. 
Pick an arbitrary edge $e\in P_1\cap P_2$. 
Let $P^1_1$ and $P^2_1$ be the path segments along $P_1$ before and after $e$ respectively. Similarly, we can define $P^1_2$ and $P^2_2$. 
Construct $P_3 = P^1_2 \cup \{e\} \cup P^2_1$ and $P_4= P^1_1 \cup \{e\} \cup P^2_2$. 
Hence, it is easy to see that $(P_3,P_4)\in \C{P}_{aq}\times\C{P}_{pb}$ and $t(P_1)t(P_2)=t(P_3)t(P_4)$.

Now assume $P_1 \cap P_2 = \emptyset$. 
By contradiction, assume there exists $(P_3,P_4)\in \C{P}_{aq}\times\C{P}_{pb}$ such that $t(P_1)t(P_2)=t(P_3)t(P_4)$. 
Clearly, $P_1 \cup P_2 = P_3 \cup P_4$. Then, there exist $e,e'\in P_4$ such that $head(e)=tail(e')$ and $e\in P_1$, $e'\in P_2$. 
Hence, $x_{ee'} \mid t(P_3)t(P_4)$ but $x_{ee'} \nmid t(P_1)t(P_2)$, contradicting our assumption.
\end{proof}

\begin{proof}[Proof of Lemma \ref{lemma_transfer_inequal}]
Assume $m_{ab}(\B{x})m_{pq}(\B{x}) \neq m_{aq}(\B{x})m_{pb}(\B{x})$. 
Thus there exists $(P_1,P_2)\in \C{P}_{ab}\times \C{P}_{pq}$ such that for any $(P_3,P_4) \in \C{P}_{aq}\times \C{P}_{pb}$, $t(P_1)t(P_2) \neq t(P_3)t(P_4)$, or vice versa. 
By Lemma \ref{lemma_equal_term}, $P_1 \cap P_2= \emptyset$ ($P_3 \cap P_4=\emptyset$ for the other case). 
On the other hand, if there exists disjoint path pair $(P_1,P_2)\in \C{P}_{ab} \times \C{P}_{pq}$, $t(P_1)t(P_2)$ is absent from $m_{aq}(\B{x})m_{pb}(\B{x})$. 
Moreover, there is only one term in $m_{ab}(\B{x})m_{pq}(\B{x})$ which equals $t(P_1)t(P_2)$. 
Thus $t(P_1)t(P_2)$ doesn't vanish from $m_{ab}(\B{x})m_{pq}(\B{x})$. 
Hence $m_{ab}(\B{x})m_{pq}(\B{x}) \neq m_{aq}(\B{x})m_{pb}(\B{x})$. 
Similarly, the theorem holds for the other case.
\end{proof}

\section{General Form of $\B{V}_1$ \label{app_v1}}
The following lemma shows that given any full-rank matrices $\B{A}$, $\B{B}$ and $\B{C}$ as defined in $\SR{A}'_1,\SR{A}'_2,\SR{A}'_3$, we can always find a non-zero solution to (\ref{eq_v1_align_2}), and thus construct a precoding matrix $\B{B}_1$ which satisfies (\ref{eq_v1_align_2}).
\begin{lemma}
\label{lemma_get_v1_1}
Equation (\ref{eq_v1_align_2}) has a non-zero solution in $\BF^{n+1}_{2^m}[z]$ in the form of $\B{r}(z)=(1, z, z^2, \cdots, z^n)\B{F}$, where $\B{F}$ is an $(n+1)\times (n+1)$ matrix in $\BF_{2^m}$. 
Moreover, any solution to (\ref{eq_v1_align_2}) is linearly dependent on $(1,z,\cdots,z^n)\B{F}$.
\end{lemma}
\begin{proof}
Denote $\B{D}=\B{BA}$. 
First, we will prove that $\myrank(z\B{C}-\B{D})=n$. 
Let $\B{c}_i$ and $\B{d}_i$ denote the $i$th column of $\B{C}$ and $\B{D}$ respectively. 
Hence, $\B{c}_1,\cdots,\B{c}_n$ are linearly independent and so are $\B{d}_1,\cdots,\B{d}_n$. 
Assume there exist $f_1(z),\cdots,f_n(z)\in \BF_{2^m}(z)$ such that $\sum^n_{i=1}f_i(z)(z\B{c}_i-\B{d}_i)=0$. 
Without loss of generality, assume $f_i(z)=\frac{g_i(z)}{h(z)}$ for $i\in \{1,2,\cdots,n\}$, where $g_i(z),h(z)\in \BF_{2^m}[z]$. 
Thus, $\sum^n_{i=1}g_i(z)(z\B{c}_i-\B{d}_i)=0$. 
Let $k=\max_{i\in \{1,2,\cdots,n\}}\{d_{g_i}\}$ and assume $g_i(z)=\sum^k_{l=0}a_{l,i}z^l$. 
Then, it follows
\begin{flalign*}
& \sum^n_{i=1}g_i(z)(z\B{c}_i-\B{d}_i) = \sum^k_{l=0}\sum^n_{i=1}(a_{l,i}z^{l+1}\B{c}_i-a_{l,i}z^l\B{d}_i) \\
= & z^{k+1}\sum^n_{i=1}a_{k,i}\B{c}_i + \sum^{k-1}_{l=0}z^{l+1}\sum^n_{i=1}(a_{l,i}\B{c}_i-a_{l+1,i}\B{d}_i) \\
& \hspace*{10pt} - \sum^n_{i=1}a_{0,i}\B{d}_i = \B{0}
\end{flalign*}
Therefore, the following equations must hold:
\begin{flalign*}
& \sum^n_{i=1}a_{k,i}\B{c}_i=0 \hspace*{10pt} \sum^n_{i=1}a_{0,i}\B{d}_i=0 \\ 
& \sum^n_{i=1}(a_{l,i}\B{c}_i-a_{l+1,i}\B{d}_i)=0 \hspace*{15pt} \forall l\in \{0,\cdots,k-1\}
\end{flalign*}
Thus $a_{l,i}=0$ for any $i\in \{1,\cdots,n\},l\in \{0,\cdots,k\}$, implying $f_i(z)=0$. 
Hence, $\myrank(z\B{C}-\B{D})=n$. 

Then, there must be an $n\times n$ invertible submatrix in $z\B{C}-\B{D}$. 
Without loss of generality, assume this submatrix consists of the top $n$ rows of $z\B{C}-\B{D}$ and denote this submatrix by $\B{E}_{n+1}$. 
Let $\B{b}$ denote the $(n+1)$th row of $z\B{C}-\B{D}$. 
In order to get a non-zero solution to equation (\ref{eq_v1_align_2}), we first fix $r_{n+1}(z)=-1$. 
Therefore, equation (\ref{eq_v1_align_2}) is transformed into $(r_1(z),\cdots,r_n(z))\B{E}_{n+1}=\B{b}$.
For $i\in \{1,2,\cdots,n\}$, let $\B{E}_i$ denote the submatrix acquired by replacing the $i$th row of $\B{E}_{n+1}$ with $\B{b}$. 
Hence, we get a non-zero solution to (\ref{eq_v1_align_2}):
\begin{flalign*}
\B{r}(z)=(\frac{\det\B{E_1}}{\det\B{E}_{n+1}}, \cdots, \frac{\det\B{E}_n}{\det\B{E}_{n+1}},-1)
\end{flalign*}
Moreover, $\bar{\B{r}}(z)=(\det\B{E}_1,\cdots,\det\B{E}_n,-\det\B{E}_{n+1})$ is also a solution. 
Also note that the degree of $z$ in each $\det\B{E}_i$ ($i\in \{1,2,\cdots,n+1\}$) is at most $n$. 
Thus, $\bar{\B{r}}(z)$ can be formulated as $(1,z,z^2,\cdots,z^n)\B{F}$, where $\B{F}$ is an $(n+1)\times (n+1)$ matrix in $\BF_{2^m}$. 
Since $\myrank(z\B{C}-\B{D})=n$, all the solutions to equation (\ref{eq_v1_align_2}) form a one-dimensional linear space. 
Thus, all solutions must be linearly dependent on $\bar{\B{r}}(z)$.
\end{proof}

\begin{proof}[Proof of Lemma \ref{lemma_v1}]
Let $\B{r}_i$ be the $i$th row of $\B{V}_1$, which satisfies equation (\ref{eq_v1_align}). 
According to Lemma \ref{lemma_get_v1_1}, $\B{r}_i$ must have the form $f_i(\eta(\B{x}^i))(1,\eta(\B{x}^i),\cdots,\eta^n(\B{x}^i))\B{F}$, where $f_i(z)$ is a non-zero rational function in $\BF_{2^m}(z)$. 
Hence, $\B{V}_1$ can be written as $\B{G}\B{V}^*_1\B{F}$.

According to Lemma \ref{lemma_get_v1_1}, equation (\ref{eq_v1_align_2}) can be rewritten as follows:
\begin{flalign*}
(z,z^2,\cdots,z^{n+1})\B{FC} = (1,z,\cdots,z^n)\B{FBA}
\end{flalign*}
The right side of the above equation contains no $z^{n+1}$, and thus the $(n+1)$th row of $\B{FC}$ must be zero. 
Similarly, there is no constant term on the left side of the above equation, implying that the 1st row of $\B{FBA}$ is zero.
\end{proof}

\section{Results on Multivariate Polynomial\label{app_polynomial}}
Let $\B{y}=(y_1,y_2,\cdots,y_k)$ be a vector of variables. 
For any $i\in \{1,2,\cdots,k\}$, define $\B{y}_i=(y_1,\cdots,y_{i-1},y_{i+1},\cdots,y_k)$, i.e., the vector consisting of all variables in $\B{y}$ other than $y_i$. 
Note that any polynomial $f(\B{y}) \in \BF[\B{y}]$ can be formulated as
\begin{flalign*}
f(\B{y}) = f_0(\B{y}_i) + f_1(\B{y}_i)y_i + \cdots + f_p(\B{y}_i)y^p_i
\end{flalign*}
where $f_j(\B{y}_i)\in \BF[\B{y}_i]$ for $j\in \{0,1,\cdots,p\}$ and $f_p(\B{y}_i)\neq 0$. 
Let $\BF(\B{y}_i)$ denote the field consisting of all rational functions in the form of $\frac{u(\B{y}_i)}{v(\B{y}_i)}$, where $u(\B{y}_i), v(\B{y}_i) \in \BF[\B{y}_i]$. 
Because $\BF[\B{y}_i]$ is a subset of $\BF(\B{y}_i)$, $f(\B{y})$ can also be viewed as a univariate polynomial in the ring $\BF(\B{y}_i)[y_i]$. 
For any $h(\B{y})\in \BF[\B{y}]$, we use $h(y_i)$ to denote its equivalent counterpart in $\BF(\B{y}_i)[y_i]$. 
To differentiate these two concepts, we reserve the notations, such as ``$\mid$'', ``$\mygcd$'' and ``$\mylcm$''\footnote{We use $\mylcm(f(x),g(x))$ to denote the least common multiple of two polynomials $f(x)$ and $g(x)$.}, for field $\BF$, and append ``1'' as a subscript to these notations to suggest they are specific to field $\BF(\B{y}_i)$. 
For example, for $f(\B{y}),g(\B{y})\in \BF[\B{y}]$ and $u(y_i),v(y_i)\in \BF(\B{y}_i)[y_i]$, $g(\B{y}) \mid f(\B{y})$ means that there exists $h(\B{y})\in \BF[\B{y}]$ such that $f(\B{y})=h(\B{y})g(\B{y})$, and $u(y_i)\mid_1 v(y_i)$ means that there exists $w(y_i)\in \BF[\B{y}_i](y_i)$ such that $v(y_i)=w(y_i)u(y_i)$. 
Similarly, $\mygcd(f(\B{y}),g(\B{y}))$ is the greatest common divisor of $f(\B{y})$ and $g(\B{y})$ within $\BF[\B{y}]$, and $\mygcd_1(u(y_i),v(y_i))$ is the greatest common divisor of $u(y_i)$ and $v(y_i)$ in $\BF(\B{y}_i)[y_i]$.

In general, each polynomial $h(y_i) \in \BF(\B{y}_i)[y_i]$ is of the following form
\begin{flalign*}
h(y_i) = \frac{a_0(\B{y}_i)}{b_0(\B{y}_i)} + \frac{a_1(\B{y}_i)}{b_1(\B{y}_i)}y_i + \cdots + \frac{a_p(\B{y}_i)}{b_p(\B{y}_i)}y^p_i
\end{flalign*}
In the above formula, for any $j\in \{0,1,\cdots,p\}$, $a_j(\B{y}_i),b_j(\B{y}_i)\in \BF[\B{y}_i]$, $b_j(\B{y}_i)\neq 0$, $\mygcd(a_j(\B{y}_i),b_j(\B{y}_i))=1$, and $a_p(\B{y}_i)\neq 0$. 
Note that for any $y^j_i$ which is absent in $h(y_i)$, we let $a_j(\B{y}_i)=0$ and $b_j(\B{y}_i)=1$. Define the following polynomial
\begin{flalign*}
\mu_h(\B{y}_i) = \mylcm(b_0(\B{y}_i), b_1(\B{y}_i), \cdots, b_p(\B{y}_i))
\end{flalign*}
Thus, $\mu_h(\B{y}_i)\in \BF[\B{y}_i]$ and $\mu_h(\B{y}_i)h(y_i)\in \BF[\B{y}]$.

\begin{lemma}
\label{lemma_partial_div}
Assume $g(\B{y}_i)\in \BF[\B{y}_i]$ and $f(\B{y}) \in \BF[\B{y}]$ is of the form $f(\B{y})=\sum^p_{j=0}f_j(\B{y}_i)y^j_i$, where $f_j(\B{y}_i)\in \BF[\B{y}_i]$. 
Then $g(\B{y}_i) \mid f(\B{y})$ if and only if $g(\B{y}_i) \mid f_j(\B{y}_i)$ for any $j\in \{0,1,\cdots,p\}$.
\end{lemma}
\begin{proof}
Apparently, if $g(\B{y}_i) \mid f_j(\B{y}_j)$ for any $j\in \{0,1,\cdots,p\}$, $g(\B{y}_i) \mid f(\B{y})$. 
Now assume $g(\B{y}_i) \mid f(\B{y})$. 
Thus there exists $h(\B{y})\in \BF[\B{y}]$ such that $f(\B{y})=g(\B{y}_i)h(\B{y})$. 
Let $h(\B{y})=\sum^p_{j=0}h_j(\B{y}_i)y^j_i$. 
Hence, it follows that $f_j(\B{y}_i)=h_j(\B{y}_i)g(\B{y}_i)$ and thus $g(\B{y}_i) \mid f_j(\B{y}_i)$.
\end{proof}

The following result follows immediately from Lemma \ref{lemma_partial_div}.
\begin{corollary}
\label{cor_partial_div}
Let $g(\B{y}_i)$ and $f(\B{y})$ be defined as Lemma \ref{lemma_partial_div}. 
Then $\gcd(g(\B{y}_i), f(\B{y})) = \gcd(g(\B{y}_i), f_0(\B{y}_i), \cdots, f_p(\B{y}_i))$.
\end{corollary}
\begin{proof}
Note that any divisor of $g(\B{y}_i)$ must be a polynomial in $\BF[\B{y}_i]$. 
Let $d(\B{y}_i)=\gcd(g(\B{y}_i),f(\B{y}))$ and $d'(\B{y}_i)=\gcd(g(\B{y}_i), f_0(\B{y}_i), \cdots, f_p(\B{y}_i))$. 
By Lemma \ref{lemma_partial_div}, $d(\B{y}_i)\mid f_j(\B{y}_i)$ for any $j\in \{0,1,\cdots,p\}$, implying that $d(\B{y}_i) \mid d'(\B{y}_i)$. 
On the other hand, $d'(\B{y}_i) \mid f(\B{y})$, and thus $d'(\B{y}_i) \mid d(\B{y}_i)$. 
Hence, $d(\B{y}_i)=d'(\B{y}_i)$.
\end{proof}

\begin{corollary}
\label{cor_partial_div_2}
For $t\in \{1,2,\cdots,s\}$, let $f_t(\B{y})\in \BF[\B{y}]$ be defined as $f_t(\B{y})=\sum^{p_t}_{j=0}f_{tj}(\B{y}_i)y^j_i$, where $f_{tj}(\B{y}_i) \in \BF[\B{y}_i]$. 
Let $g(\B{y}_i)\in \BF[\B{y}_i]$. 
It follows
\begin{flalign*}
& \mygcd(g(\B{y}_i), f_1(\B{y}), \cdots, f_t(\B{y})) \\
=& \mygcd(g(\B{y}_i), f_{10}(\B{y}_i), \cdots, f_{1p_1}(\B{y}_i), \cdots,\\ 
& \hspace*{40pt} f_{s0}(\B{y}_i), \cdots, f_{sp_s}(\B{y}_i))
\end{flalign*}
\end{corollary}
\begin{proof}
We have the following equations
\begin{flalign*}
& \mygcd(g(\B{y}_i), f_1(\B{y}), \cdots, f_t(\B{y})) \\
=& \mygcd(g(\B{y}_i), f_1(\B{y}), \cdots, g(\B{y}_i), f_t(\B{y})) \\
=& \mygcd(\mygcd(g(\B{y}_i), f_1(\B{y})), \cdots, \mygcd(g(\B{y}_i), f_s(\B{y}))) \\
=& \mygcd(g(\B{y}_i), f_{10}(\B{y}_i), \cdots, f_{1p_1}(\B{y}_i), \cdots, \\
& \hspace*{40pt} g(\B{y}_i), f_{s0}(\B{y}_i), \cdots, f_{sp_s}(\B{y}_i)) \\
=& \mygcd(g(\B{y}_i), f_{10}(\B{y}_i), \cdots, f_{1p_1}(\B{y}_i), \cdots,\\ 
& \hspace*{40pt} f_{s0}(\B{y}_i), \cdots, f_{sp_s}(\B{y}_i))
\end{flalign*}
\end{proof}

\begin{lemma}
\label{lemma_relative_prime_1}
For $t\in \{1,2,\cdots,s\}$, let $a_t(\B{y}),b_t(\B{y})\in \BF[\B{y}]$ such that $b_t(\B{y})\neq 0$ and $\mygcd(a_t(\B{y}), b_t(\B{y}))=1$.
For $t\in \{1,2,\cdots,s\}$, let $v_t(\B{y})=\mylcm(b_1(\B{y}),\cdots,b_t(\B{y}))$.
Then we have
\begin{flalign*}
\mygcd\left(a_1(\B{y})\frac{v_s(\B{y})}{b_1(\B{y})}, \cdots, a_s(\B{y})\frac{v_s(\B{y})}{b_s(\B{y})}, v_s(\B{y})\right) = 1
\end{flalign*}
\end{lemma}
\begin{proof}
We use induction on $s$ to prove this lemma. Apparently, the lemma holds for $s=1$ due to $\mygcd(a_1(\B{y}),b_1(\B{y}))=1$. Assume it holds for $s-1$. Thus it follows
\begin{flalign*}
& \mygcd\bigg(a_1(\B{y})\frac{v_s(\B{y})}{b_1(\B{y})}, \cdots, a_s(\B{y})\frac{v_s(\B{y})}{b_s(\B{y})}, v_s(\B{y})\bigg) \\
=& \mygcd\bigg(a_1(\B{y})\frac{v_s(\B{y})}{b_1(\B{y})}, \cdots, a_s(\B{y})\frac{v_s(\B{y})}{b_s(\B{y})}, b_s(\B{y})\frac{v_s(\B{y})}{b_s(\B{y})}\bigg) \\
=& \mygcd\bigg(a_1(\B{y})\frac{v_s(\B{y})}{b_1(\B{y})}, \cdots,\mygcd(a_s(\B{y}), b_s(\B{y})) \frac{v_s(\B{y})}{b_s(\B{y})}\bigg) \\
\overset{(a)}{=} & \mygcd\bigg(a_1(\B{y})\frac{v_s(\B{y})}{b_1(\B{y})}, \cdots, a_{s-1}(\B{y})\frac{v_s(\B{y})}{b_{s-1}(\B{y})}, \frac{v_s(\B{y})}{b_s(\B{y})}\bigg) \\
\overset{(b)}{=} & \mygcd\bigg(a_1(\B{y})\frac{v_s(\B{y})}{b_1(\B{y})}, \cdots, a_{s-1}(\B{y})\frac{v_s(\B{y})}{b_{s-1}(\B{y})},  \\
& \hspace*{3cm} \mygcd\bigg(v_{s-1}(\B{y}),\frac{v_s(\B{y})}{b_s(\B{y})}\bigg) \bigg) \\
=& \mygcd\bigg(a_1(\B{y})\frac{v_s(\B{y})}{b_1(\B{y})}, \cdots, a_{s-1}(\B{y})\frac{v_s(\B{y})}{b_{s-1}(\B{y})}, v_{s-1}(\B{y}),\frac{v_s(\B{y})}{b_s(\B{y})}\bigg) \\
=& \mygcd\bigg( \frac{v_s(\B{y})}{v_{s-1}(\B{y})}\mygcd\bigg(a_1(\B{y})\frac{v_{s-1}(\B{y})}{b_1(\B{y})}, \cdots, a_{s-1}(\B{y})\frac{v_{s-1}(\B{y})}{b_{s-1}(\B{y})}\bigg), \\
& \hspace*{3cm} v_{s-1}(\B{y}),\frac{v_s(\B{y})}{b_s(\B{y})}\bigg) \\
\overset{(c)}{=}& \mygcd\bigg(\frac{v_s(\B{y})}{v_{s-1}(\B{y})}, v_{s-1}(\B{y}), \frac{v_s(\B{y})}{b_s(\B{y})}\bigg) \\
\overset{(d)}{=} & \mygcd\bigg(\frac{b_s(\B{y})}{\mygcd(v_{s-1}(\B{y}),b_s(\B{y}))}, v_{s-1}(\B{y}), \frac{v_{s-1}(\B{y})}{\mygcd(v_{s-1}(\B{y}),b_s(\B{y}))} \bigg) \\
=& \mygcd(1, v_{s-1}(\B{y})) = 1
\end{flalign*}
In the above equations, (a) is due to $\mygcd(a_s(\B{y}),b_s(\B{y}))=1$; 
(b) follows from the fact that $\frac{v_s(\B{y})}{b_s(\B{y})} \mid v_{s-1}(\B{y})$ and thus $\frac{v_s(\B{y})}{b_s(\B{y})} = \mygcd(v_{s-1}(\B{y}), \frac{v_s(\B{y})}{b_s(\B{y})})$; 
(c) follows from the inductive assumption; 
(d) is due to the equality: $v_s(\B{y})=\mylcm(v_{s-1}(\B{y}),b_s(\B{y}))=\frac{v_{s-1}(\B{y})b_s(\B{y})}{\mygcd(v_{s-1}(\B{y}), b_s(\B{y}))}$.
\end{proof}

\begin{corollary}
\label{cor_relative_prime_2}
For $j\in \{1,2,\cdots,s\}$, let $f_j(y_i) \in \BF(\B{y}_i)[y_i]$. 
Define $v(\B{y}_i)=\mylcm(\mu_{f_1}(\B{y}_i), \cdots, \mu_{f_s}(\B{y}_i))$ and $\bar{f}_j(\B{y})=v(\B{y}_i)f_j(y_i)$. 
Thus  $\mygcd(v(\B{y}_i), \bar{f}_1(\B{y}), \cdots, \bar{f}_s(\B{y}))=1$
\end{corollary}
\begin{proof}
Assume $f_j(y_i)$ has the following form:
\begin{flalign*}
f_j(y_i) = \frac{a_{j0}(\B{y}_i)}{b_{j0}(\B{y}_i)} + \frac{a_{j1}(\B{y}_i)}{b_{j1}(\B{y}_i)}y_i + \cdots + \frac{a_{jp_j}(\B{y}_i)}{b_{jp_j}(\B{y}_i)}y^{p_j}_i
\end{flalign*}
where for any $j\in \{1,2,\cdots,s\}$ and $t\in \{0,1,\cdots,p_j\}$, $a_{jt}(\B{y}_i), b_{jt}(\B{y}_i)\in \BF[\B{y}_i]$, $b_{jt}(\B{y}_i)\neq 0$ and $\mygcd(a_{jt}(\B{y}_i),b_{jt}(\B{y}_i))=1$. 
Apparently, $v(\B{y}_i)$ is the least common multiple of all $b_{jt}(\B{y}_i)$'s. 
Define $u_{jt}(\B{y}_i)=\frac{v(\B{y}_i)}{b_{jt}(\B{y}_i)} \in \BF[\B{y}_i]$. 
Hence, we have $\bar{f}_j(\B{y}) = \sum^{p_j}_{t=0}a_{jt}(\B{y}_i)u_{jt}(\B{y}_i)y^t_i$. 
Then it follows
\begin{flalign*}
&\mygcd(v(\B{y}_i), \bar{f}_1(\B{y})), \cdots, \bar{f}_s(\B{y})) \\
\overset{(a)}{=}& \mygcd(v(\B{y}_i), a_{10}(\B{y}_i)u_{10}(\B{y}_i), \cdots, a_{1p_1}(\B{y}_i)u_{1p_1}(\B{y}_i), \cdots, \\
& \hspace*{30pt} a_{s0}(\B{y}_i)u_{s0}(\B{y}_i), \cdots, a_{sp_s}(\B{y}_i)u_{sp_s}(\B{y}_i)) \\
\overset{(b)}{=}& 1
\end{flalign*}
where (a) is due to Corollary \ref{cor_partial_div_2} and (b) follows from Lemma \ref{lemma_relative_prime_1}.
\end{proof}

Generally, the definitions of division in $\BF[\B{y}]$ and $\BF(\B{y}_i)[y_i]$ are different. 
However, the following theorem reveals the two definitions are closely related.
\begin{theorem}
\label{th_multivar_div}
Consider two polynomials $f(\B{y}),g(\B{y})\in \BF[\B{y}]$, where $g(\B{y})\neq 0$. 
Then $g(\B{y})\mid f(\B{y})$ if and only if $g(y_i) \mid_1 f(y_i)$ for every $i\in \{1,2,\cdots,k\}$.
\end{theorem}
\begin{proof}
The division equation between $f(y_i)$ and $g(y_i)$ is as follows
\begin{flalign}
\label{eq_div}
f(y_i) = h_i(y_i)g(y_i) + r_i(y_i)
\end{flalign}
where $h_i(y_i), r_i(y_i) \in \BF(\B{y}_i)[y_i]$, and either $r_i(y_i)=0$ or $d_{r_i} < d_g$. 
Due to the uniqueness of Equation (\ref{eq_div}), $f(\B{y})\mid g(\B{y})$ immediately implies that for any $i\in \{1,2,\cdots,k\}$, $r_i(y_i)=0$ and thus $g(y_i) \mid_1 f(y_i)$.

Conversely, assume for every $i\in \{1,\cdots,k\}$, $g(y_i) \mid_1 f(y_i)$ and hence $r_i(y_i)= 0$. 
Denote $\bar{h}_i(\B{y})=\mu_{h_i}(\B{y}_i)h_i(y_i)$. 
Clearly, $\bar{h}_i(\B{y}) \in \BF[\B{y}]$. 
Then, the following equation holds
\begin{flalign*}
\mu_{h_i}(\B{y}_i)f(\B{y}) = \bar{h}_i(\B{y})g(\B{y})
\end{flalign*}
By Corollary \ref{cor_relative_prime_2}, $\mygcd(\mu_{h_i}(\B{y}_i), \bar{h}_i(\B{y}))=1$. 
Thus, $\mu_{h_i}(\B{y}_i) \mid g(\B{y})$. 
Define $\bar{g}(\B{y})=\frac{g(\B{y})}{\mu_{h_i}(\B{y}_i)}$. 
By Lemma \ref{lemma_partial_div}, $\bar{g}(\B{y})\in \BF[\B{y}]$. 
Define $u(\B{y})=\frac{g(\B{y})}{\mygcd(f(\B{y}),g(\B{y}))} \in \BF[\B{y}]$. 
It follows that
\begin{flalign*}
u(\B{y}) &= \frac{g(\B{y})}{\mygcd(f(\B{y}),g(\B{y}))} \\
&= \frac{\mu_{h_i}(\B{y}_i)\bar{g}(\B{y})}{\mygcd(\bar{h}_i(\B{y})\bar{g}(\B{y}), \mu_{h_i}(\B{y}_i)\bar{g}(\B{y}))} \\
&= \frac{\mu_{h_i}(\B{y}_i)\bar{g}(\B{y})}{\bar{g}(\B{y})\mygcd(\bar{h}_i(\B{y}), \mu_{h_i}(\B{y}_i))} \\
&= \frac{\mu_{h_i}(\B{y}_i)\bar{g}(\B{y})}{\bar{g}(\B{y})} \\
&= \mu_{h_i}(\B{y}_i)
\end{flalign*}
Note that variable $y_i$ is absent in $u(\B{y})$. 
Because $y_i$ can be any arbitrary variable in $\B{y}$, it immediately follows that all the variables in $\B{y}$ must be absent in $u(\B{y})$, implying that $u(\B{y})$ is a constant in $\BF$. 
Hence $g(\B{y})\mid f(\B{y})$.
\end{proof}

\begin{theorem}
\label{th_multivar_gcd}
Let $f(\B{y}),g(\B{y})$ be two non-zero polynomials in $\BF[\B{y}]$. 
Then $\mygcd(f(\B{y}),g(\B{y}))=1$ if and only if $\mygcd_1(f(y_i),g(y_i))=1$ for any $i\in \{1,2,\cdots,k\}$.
\end{theorem}
\begin{proof}
First, assume for any $i\in \{1,2,\cdots,k\}$, $\mygcd_1(f(y_i),g(y_i))=1$. 
We use contradiction to prove that $\mygcd(f(\B{y}),g(\B{y}))=1$. 
Assume $u(\B{y})=\mygcd(f(\B{y}),g(\B{y}))$ is not constant. 
Let $y_i$ be a variable which is present in $u(\B{y})$. 
By Theorem \ref{th_multivar_div}, $u(y_i)\mid_1 f(y_i)$ and $u(y_i)\mid_1 g(y_i)$, which contradicts that $\mygcd_1(f(y_i),g(y_i))=1$.

Then, assume $\mygcd(f(\B{y}),g(\B{y}))=1$. 
We also use contradiction to prove that for any $i\in \{1,2,\cdots,k\}$, $\mygcd_1(f(y_i),g(y_i))=1$. 
Assume there exists $i\in \{1,\cdots,k\}$ such that $v(y_i)=\mygcd_1(f(y_i),g(y_i))$ is non-trivial. 
Define $w(\B{y})=\mu_v(\B{y}_i)v(y_i) \in \BF[\B{y}]$. 
Clearly, $w(y_i)\mid_1 f(y_i)$ and $w(y_i)\mid_1 g(y_i)$. 
Thus, there exists $p(y_i),q(y_i) \in \BF(\B{y}_i)[y_i]$ such that
\begin{flalign*}
f(y_i) = w(y_i)p(y_i) \hspace*{20pt} g(y_i) = w(y_i)q(y_i)
\end{flalign*}
Let $s(\B{y}_i)=\mylcm(\mu_p(\B{y}_i), \mu_q(\B{y}_i))$. 
Define $\bar{p}(\B{y})=s(\B{y}_i)p(y_i)$ and $\bar{q}(\B{y})=s(\B{y}_i)q(y_i)$. 
Apparently, $\bar{p}(\B{y}),\bar{q}(\B{y})\in \BF[\B{y}]$. 
It follows that
\begin{flalign*}
s(\B{y}_i)f(\B{y}) = w(\B{y})\bar{p}(\B{y}) \hspace*{20pt} s(\B{y}_i)g(\B{y}) = w(\B{y})\bar{q}(\B{y})
\end{flalign*}
Then the following equation holds
\begin{flalign*}
s(\B{y}_i)\mygcd(f(\B{y}), g(\B{y})) = w(\B{y})\mygcd(\bar{p}(\B{y}), \bar{q}(\B{y}))
\end{flalign*}
Due to Corollary \ref{cor_relative_prime_2}, $\mygcd(s(\B{y}_i), \mygcd(\bar{p}(\B{y}), \bar{q}(\B{y}))) = \mygcd(s(\B{y}_i), \bar{p}(\B{y}), \bar{q}(\B{y})) =1$. 
Hence $s(\B{y}_i) \mid w(\B{y})$. Let $\bar{w}(\B{y})=\frac{w(\B{y})}{s(\B{y}_i)}$. 
According to Lemma \ref{lemma_partial_div}, $\bar{w}(\B{y})$ is a non-trivial polynomial in $\BF[\B{y}]$. 
Thus, $\bar{w}(\B{y}) \mid \mygcd(f(\B{y}), g(\B{y}))$, contradicting $\mygcd(f(\B{y}), g(\B{y}))=1$.
\end{proof}

\begin{lemma}
\label{lemma_relative_prime_univar}
Consider two non-zero polynomials in $\BF[z]$, $f(z) = a_0 + a_1z + \cdots + a_pz^p$ and $g(z) = b_0 + b_1z + \cdots + b_qz^q$, where $a_i,b_j\in \BF$ for $i\in\{0,1,\cdots,p\},j\in \{0,1,\cdots,q\}$, $a_pb_q \neq 0$, $p\ge q$ and $\mygcd(f(z),g(z))=1$. Let $s(x), t(x)$ be two non-zero polynomials in $\BF[x]$ such that $\mygcd(s(x),t(x))=1$. Define the following polynomials in $\BF[x]$: 
\begin{align*}
\alpha(x)=f\bigg(\frac{s(x)}{t(x)}\bigg)t^p(x) =\sum^p_{k=0}\nolimits a_kt^{p-k}(x)s^k(x) \\
\beta(x)=g\bigg(\frac{s(x)}{t(x)}\bigg)t^p(x) =\sum^q_{k=0}\nolimits b_kt^{p-k}(x)s^k(x)
\end{align*}
Then $\mygcd(\alpha(x),\beta(x))=1$.
\end{lemma}
\begin{proof}
Assume $w(x)=\mygcd(\alpha(x),\beta(x))$ is non-trivial. 
Thus we can find an extension field $\bar{\BF}$ of $\BF$ such that there exists $x_0 \in \bar{\BF}$ which satisfies $w(x_0)=0$ and hence $\alpha(x_0)=\beta(x_0)=0$. 
In the rest of this proof, we restrict our discussion in $\bar{\BF}$. 
Note that $\mygcd(f(z),g(z))=1$ and $\mygcd(s(x),t(x))=1$ also hold for $\bar{\BF}$.  
Assume $t(x_0)= 0$ and thus $x-x_0 \mid t(x)$. 
Since $\mygcd(s(x),t(x))=1$, it follows that $x-x_0 \nmid s(x)$ and thus $s(x_0)\neq 0$. 
Hence, $\alpha(x_0)=a_ps^p(x_0)\neq 0$, contradicting that $\alpha(x_0)=0$. 
Hence, we have proved that $t(x_0)\neq 0$. 
Then we have
\begin{flalign*}
f\bigg(\frac{s(x_0)}{t(x_0)}\bigg) = \frac{\alpha(x_0)}{t^p(x_0)} = 0 \hspace*{10pt} g\bigg(\frac{s(x_0)}{t(x_0)}\bigg) = \frac{\beta(x_0)}{t^p(x_0)} = 0
\end{flalign*}
which implies that $z-\frac{s(x_0)}{t(x_0)}$ is a common divisor of $f(z)$ and $g(z)$, contradicting $\mygcd(f(z),g(z))=1$.
\end{proof}

\begin{proof}[Proof of Lemma \ref{lemma_relative_prime_multivar}]
Note that if we substitute $\BF$ with $\BF(\B{y}_i)$ and $\mygcd$ with $\mygcd_1$ in Lemma \ref{lemma_relative_prime_univar}, the lemma also holds. 
Apparently, $f(z),g(z)\in \BF(\B{y}_i)[z]$. We will prove that $\mygcd_1(f(z),g(z))=1$. 
By contradiction, assume $r(z)=\mygcd_1(f(z),g(z))\in \BF(\B{y}_i)[z]$ is non-trivial. Let $\bar{f}(z)=\frac{f(z)}{r(z)}$ and $\bar{g}(z)=\frac{g(z)}{r(z)}$. 
Clearly, $\bar{f}(z)$ and $\bar{g}(z)$ are both non-zero polynomials in $\BF(\B{y}_i)[z]$. 
Then we can find an assignment to $\B{y}_i$, denoted by $\B{y}^*_i$, such that the coefficients of the maximum powers of $z$ in $r(z),\bar{f}(z)$ and $\bar{g}(z)$ are all non-zeros. 
Let $\bar{r}(z)$ denote the univariate polynomial acquired by assigning $\B{y}_i=\B{y}^*_i$ to $r(z)$. 
Clearly, $\bar{r}(z)$ is a common divisor of $f(z)$ and $g(z)$ in $\BF[z]$, contradicting $\mygcd(f(z),g(z))=1$. 
Moreover, due to $\mygcd(s(\B{y}),t(\B{y}))=1$ and Theorem \ref{th_multivar_gcd}, $\mygcd_1(s(y_i),t(y_i))=1$. 
Thus, by Lemma \ref{lemma_relative_prime_univar}, $\mygcd_1(\alpha(y_i),\beta(y_i))=1$. 
Since $i$ can be any integer in $\{1,2,\cdots,k\}$, it follows that $\mygcd(\alpha(\B{y}),\beta(\B{y}))=1$ by Theorem \ref{th_multivar_gcd}.
\end{proof}

\bibliographystyle{IEEEtran}
\bibliography{./nc-short}

\end{document}